\documentclass{article}
\usepackage{amsmath}
\usepackage{amsthm}
\usepackage{amssymb}

\usepackage[round]{natbib}

\usepackage{amsmath,amsthm,amssymb,sansmath,color,geometry,multicol,graphicx,float,color,authblk, hyperref,stmaryrd,calc}
\usepackage[english]{babel}
\usepackage[utf8]{inputenc}
\usepackage[inline]{enumitem}
\usepackage[active]{srcltx}
\usepackage{upgreek}
\usepackage{mathrsfs}
\usepackage[utf8]{inputenc}
\usepackage{amssymb}
\usepackage{makeidx}
\usepackage[english]{babel}
\usepackage{graphicx}
\usepackage{amsfonts,amssymb}
\usepackage{amsmath}
\usepackage{oldgerm}
\usepackage{mathrsfs}
\usepackage[active]{srcltx}
\usepackage{verbatim}
\usepackage{aliascnt}
\usepackage{array}
\usepackage[textwidth=4cm,textsize=footnotesize]{todonotes}
\hypersetup{
  colorlinks = true,
  citecolor = blue,
}
\usepackage{xargs}
\usepackage{cellspace}
\usepackage{algorithm,algorithmic}

\usepackage{bbm}

\usepackage{aliascnt}
\usepackage{cleveref}
\usepackage{autonum}
\makeatletter
\newtheorem{theorem}{Theorem}
\crefname{theorem}{theorem}{Theorems}
\Crefname{Theorem}{Theorem}{Theorems}

\newaliascnt{lemma}{theorem}
\newtheorem{lemma}[lemma]{Lemma}
\aliascntresetthe{lemma}
\crefname{lemma}{lemma}{lemmas}
\Crefname{Lemma}{Lemma}{Lemmas}

\newaliascnt{corollary}{theorem}

\aliascntresetthe{corollary}
\crefname{corollary}{corollary}{corollaries}
\Crefname{Corollary}{Corollary}{Corollaries}

\newaliascnt{proposition}{theorem}
\newtheorem{proposition}[proposition]{Proposition}
\aliascntresetthe{proposition}
\crefname{proposition}{proposition}{propositions}
\Crefname{Proposition}{Proposition}{Propositions}

\newaliascnt{definition}{theorem}
\newtheorem{definition}[definition]{Definition}
\aliascntresetthe{definition}
\crefname{definition}{definition}{definitions}
\Crefname{Definition}{Definition}{Definitions}

\newaliascnt{definitionProposition}{theorem}

\aliascntresetthe{definitionProposition}
\crefname{Proposition and Definition}{Proposition and Definition}{Proposition and Definition}
\Crefname{Proposition and Definition}{Proposition and Definition}{Proposition and Definition}

\newaliascnt{remark}{theorem}

\aliascntresetthe{remark}
\crefname{remark}{remark}{remarks}
\Crefname{Remark}{Remark}{Remarks}

\newtheorem{example}[theorem]{Example}
\crefname{example}{example}{examples}
\Crefname{Example}{Example}{Examples}

\crefname{algorithm}{algorithm}{algorithms}
\Crefname{Algorithm}{Algorithm}{Algorithms}

\crefname{figure}{figure}{figures}
\Crefname{Figure}{Figure}{Figures}

\Crefname{assumption}{\textbf{H}\hspace{-3pt}}{\textbf{H}\hspace{-3pt}}
\crefname{assumption}{\textbf{H}}{\textbf{H}}

\Crefname{assumptionL}{\textbf{L}\hspace{-3pt}}{\textbf{L}\hspace{-3pt}}
\crefname{assumptionL}{\textbf{L}}{\textbf{L}}

\newtheorem{assumptionL2}{\textbf{L2HMC}\hspace{-3pt}}
\Crefname{assumptionL2}{\textbf{L2HMC}\hspace{-3pt}}{\textbf{L2HMC}\hspace{-3pt}}
\crefname{assumptionL2}{\textbf{L2HMC}}{\textbf{L2HMC}}

\Crefname{assumptionG}{\textbf{G}\hspace{-3pt}}{\textbf{G}\hspace{-3pt}}
\crefname{assumptionG}{\textbf{G}}{\textbf{G}}

\newtheorem{assumptionN}{\textbf{NICE}\hspace{-3pt}}
\Crefname{assumptionN}{\textbf{NICE}\hspace{-3pt}}{\textbf{NICE}\hspace{-3pt}}
\crefname{assumptionN}{\textbf{NICE}}{\textbf{NICE}}

\def\msz{\mathsf{Z}}

\def\msv{\mathsf{V}}

\def\mcz{\mathcal{Z}}

\def\rset{\mathbb{R}}

\def\nset{\mathbb{N}}
\def\nsets{\mathbb{N}^*}

\def\Zset{\mathbb{Z}}

\def\rmd{\mathrm{d}}

\def\rme{\mathrm{e}}

\def\rmc{\mathrm{C}}
\def\rmC{\mathrm{C}}

\newcommandx{\functionspace}[2][1=+]{\mathbb{F}_{#1}(#2)}
\newcommand{\1}{\mathbbm{1}}

\newcommand{\LeftEqNo}{\let\veqno\@@leqno}

\newcommand{\N}{\ensuremath{\mathbb{N}}}

\newcommandx{\Vnorm}[2][1=V]{\| #2 \|_{#1}}
\newcommandx{\VnormEq}[2][1=V]{\left\| #2 \right\|_{#1}}
\newcommandx{\norm}[2][1=]{\ifthenelse{\equal{#1}{}}{\left\Vert #2 \right\Vert}{\left\Vert #2 \right\Vert^{#1}}}
\newcommandx{\normLigne}[2][1=]{\ifthenelse{\equal{#1}{}}{\Vert #2 \Vert}{\Vert #2\Vert^{#1}}}

\newcommand{\parentheseDeux}[1]{\left[ #1 \right]}
\newcommand{\defEns}[1]{\left\lbrace #1 \right\rbrace }
\newcommand{\defEnsLigne}[1]{\lbrace #1 \rbrace }

\newcommand{\eqdef}{\overset{:=}}

\newcommandx\probaMarkovTilde[2][2=]
{\ifthenelse{\equal{#2}{}}{{\widetilde{\mathbb{P}}_{#1}}}{\widetilde{\mathbb{P}}_{#1}\left[ #2\right]}}

\def\ie{\textit{i.e.}}

\def\eqsp{\;}

\newcommand{\ooint}[1]{\left(#1\right)}
\newcommand{\ccint}[1]{\left[#1\right]}

\newcommand{\indi}[1]{\1_{#1}}
\newcommandx{\weight}[2][2=n]{\omega_{#1,#2}^N}

\newcommandx\sequence[3][2=,3=]
{\ifthenelse{\equal{#3}{}}{\ensuremath{\{ #1_{#2}\}}}{\ensuremath{\{ #1_{#2}, \eqsp #2 \in #3 \}}}}
\newcommandx\sequenceD[3][2=,3=]
{\ifthenelse{\equal{#3}{}}{\ensuremath{\{ #1_{#2}\}}}{\ensuremath{( #1)_{ #2 \in #3} }}}

\newcommandx{\sequencen}[2][2=n\in\N]{\ensuremath{\{ #1_n, \eqsp #2 \}}}
\newcommandx\sequenceDouble[4][3=,4=]
{\ifthenelse{\equal{#3}{}}{\ensuremath{\{ (#1_{#3},#2_{#3}) \}}}{\ensuremath{\{  (#1_{#3},#2_{#3}), \eqsp #3 \in #4 \}}}}
\newcommandx{\sequencenDouble}[3][3=n\in\N]{\ensuremath{\{ (#1_{n},#2_{n}), \eqsp #3 \}}}

\newcommand{\wrt}{w.r.t.}

\def\eg{e.g.}

\newcommand{\opnorm}[1]{{\left\vert\kern-0.25ex\left\vert\kern-0.25ex\left\vert #1
    \right\vert\kern-0.25ex\right\vert\kern-0.25ex\right\vert}}

\def\Id{\operatorname{Id}}

\newcommandx{\CPE}[3][1=]{{\mathbb E}_{#1}\left[#2 \left \vert #3 \right. \right]} %%%% esperance conditionnelle
\newcommandx{\CPVar}[3][1=]{\mathrm{Var}^{#3}_{#1}\left\{ #2 \right\}}
\newcommand{\CPP}[3][]
{\ifthenelse{\equal{#1}{}}{{\mathbb P}\left(\left. #2 \, \right| #3 \right)}{{\mathbb P}_{#1}\left(\left. #2 \, \right | #3 \right)}}

\newcommandx{\osc}[2][1=]{\mathrm{osc}_{#1}(#2)}

\def\Id{\operatorname{Id}}

\def\tG{\tilde{G}}
\def\tH{\tilde{H}}

\def\Jac{\operatorname{J}}

\newcommand{\ensemble}[2]{\left\{#1\,:\eqsp #2\right\}}

\newcommand{\set}[2]{\ensemble{#1}{#2}}

\newcommand\coupling[2]{\Gamma(\mu,\nu)}

\newcommand{\comp}{\mathrm{c}}

\renewcommand{\geq}{\geqslant}
\renewcommand{\leq}{\leqslant}

\def\Leb{\mathrm{Leb}}

\def\iff{ if and only if }

\def\projp{\operatorname{proj}^{\msp}}

\def\invf{s}

\def\accfun{\mathsf{a}}

\def\transflif{\Psi}
\def\transflifdet{\Phi}
\def\transfdet{\Phi}
\def\transfm{M}
\def\transfn{N}

\def\projq{\operatorname{proj}_1}
\def\projp{\operatorname{proj}_2}

\def\Zset{\msz}
\def\Zsigma{\mathcal{Z}}
\def\invol{s}
\def\involk{S}
\def\dnu{\check{\nu}}
\def\dmu{\check{\mu}}
\def\dlambda{\check{\lambda}}
\def\drho{\check{\rho}}
\newcommand{\pushf}[2]{#1_{\#}#2}
\def\mae{\ensuremath{\text{a.e.}}}

\def\sym{F}
\def\target{\pi_0}
\def\exttarget{\pi}
\def\suppl{\textcolor{red}{\ensuremath{*}}}

\def\lipschitz{\mathrm{L}}

\def\probav{\rho}
\def\LF{\operatorname{F}}
\def\LG{\operatorname{G}}
\def\partialrefresh{\omega}
\def\Uniform{\mathcal{U}}
\def\Ber{\operatorname{Ber}}

\title{Nonreversible MCMC from conditional invertible transforms: a complete recipe with convergence guarantees}

\author{Achille Thin\textsuperscript{1},  Nikita Kotelevskii\textsuperscript{2}, Christophe Andrieu\textsuperscript{3}, Alain Durmus\textsuperscript{4} \and Eric Moulines\textsuperscript{1}\textsuperscript{5}, Maxim Panov\textsuperscript{2}}

\date{}

\begin{document}

\footnotetext[1]{Centre de Math\'ematiques Appliqu\'ees, UMR 7641, Ecole Polytechnique, France. \\
  achille.thin@polytechnique.edu, eric.moulines@polytechnique.edu}
\footnotetext[2]{CDISE, Skolkovo Institute of Science and Technology, Moscow, Russia\\
Nikita.Kotelevskii@skoltech.ru,  M.Panov@skoltech.ru
}
\footnotetext[3]{School of Mathematics,  University of Bristol, UK.
  c.andrieu@bristol.ac.uk
  }
\footnotetext[4]{Université Paris-Saclay, ENS Paris-Saclay, CNRS, Centre Borelli, F-91190 Gif-sur-Yvette, France.\\
alain.durmus@cmla.ens-cachan.fr}
\footnotetext[5]{CS Departement, HSE University, Russian Federation}

\maketitle

% \aistatstitle{
% %\color{red} Nonreversible MCMC from conditional invertible transforms: theory and methods}\\
% }

% \aistatsauthor{ Achille Thin \And Nikita Kotolevskii \And Maxim Panov \And Christophe Andrieu \And  Alain Durmus \And Eric Moulines}

% \aistatsaddress{ Ecole polytechnique, Paris \And  National Research University, Higher School of Economics \And Skolkovo Institute of Science and Technology \And University of Bristol } ]

\begin{abstract}
  Markov Chain Monte Carlo (MCMC) is a class of algorithms to sample complex and high-dimensional probability distributions. The Metropolis-Hastings (MH) algorithm, the workhorse of MCMC, provides a simple recipe to construct reversible Markov kernels. Reversibility is a tractable property that implies a less tractable but essential property here, invariance. Reversibility is however not necessarily desirable when considering performance. This has prompted recent interest in designing kernels breaking this property. At the same time, an active stream of research has focused on the design of novel versions of the MH kernel, some nonreversible, relying on the use of complex invertible deterministic transforms. While standard implementations of the MH kernel are well understood, the aforementioned developments have not received the same systematic treatment to ensure their validity. This paper fills the gap by developing general tools to ensure that a class of nonreversible Markov kernels, possibly relying on complex transforms, has the desired invariance property and leads to convergent algorithms. This leads to a set of simple and practically verifiable conditions. % Numerical illustrations are presented to illustrate our  findings.

\end{abstract}

\section{Introduction}

Being able to simulate from a probability distribution, say $\pi$ defined on a measurable space $(\Zset,\mcz)$ and referred to as the target distribution hereafter, is a ubiquitous task.
% in science  where applications range from the simulation of models in physics to posterior distribution simulation in Bayesian inference.
Markov chain Monte Carlo methods (MCMC) is an important body of versatile techniques to sample from $\pi$. They consist of simulating realisations of time-homogeneous Markov chains $(Z_k)_{k \in\nset}$ of invariant distribution $\pi$ which possess the property that their realised states can be used to mimic samples from $\pi$, that is $Z_k \sim \pi$ approximately, but with arbitrary precision, and approximate expectations with respect to $\pi$ -- more precise statements are provided in \Cref{theo:basic-cv-textbook} and we refer to these, for now, lose concepts as ``convergence".  We  denote by $P$ the Markov kernel associated with $(Z_k)_{k \in\nset}$.

 Metropolis-Hastings (MH) is a popular strategy to design such a Markov kernel. In its most common form, the ``textbook" MH kernel samples the $(k+1)$-th state $Z_{k+1}$ of $(Z_k)_{k \in \nset}$ as follows:
\begin{enumerate*}[label=(\arabic*)]
\item sample a proposal $Y_{k+1} \sim Q(Z_{k},\cdot)$;
\item set $Z_{k+1} = Y_{k+1}$ with probability $\alpha(Z_{k},Y_{k+1})$; otherwise, set $Z_{k+1}= Z_k$,
\end{enumerate*}
where $Q\colon \Zset \times \Zsigma \to \ccint{0,1}$ is a Markov kernel and $\alpha\colon \Zset \times \Zset \to \ccint{0,1}$ is the acceptance probability. General conditions on $\pi,Q$ and $\alpha$ in order to ensure invariance and convergence of $(Z_k)_{k\in\nset}$ have been known for some time.
% and revolve in particular around the notion of reversibility -- this is further discussed below (see also~\cite{tierney:1998} and the supplementary for a general reversible MH kernel).
In the particular situation where $\pi$ and $\{Q(z,\cdot), z \in \Zset\}$ have densities $\pi$ and $\{q(z,\cdot),z\in \Zset\}$ with respect to a common dominating measure and are positive everywhere one can choose $\alpha(z,z')=\min\{1,\pi(z')q(z',z)/[\pi(z)q(z,z')]\}$ and define a convergent algorithm.

\textbf{Contribution \#1: a complete recipe for $(\pi,S)-$reversible
  kernels.}  In the context of MCMC the $\pi-$invariance property of
$P$ is traditionally the consequence of a stronger property,
$\pi-$reversibility (related to detailed balance~\cite{fang:2014}),
which is however more tractable in practice. The MH Markov kernel is designed to
satisfy this property. However, there has been a re-kindled interest in the
development of ``nonreversible" algorithms
\cite{turitsyn:2011,hukushima2013irreversible,ma2016unifying,ottobre2016markov,bierkens:2017,neklyudov:welling:vetrov:2020,sherlock:2019,gustafson:1998}
which come with the promise of removing the backtracking behaviour
of reversible algorithms, and hence speed-up convergence. Our first
contribution (Section~\ref{sec:mu-s-reversibility}) is \textbf{(a)} a
review of $(\pi,\involk)-$reversibility, related to the modified
detailed balance condition~\cite{fang:2014}, a generalisation of reversibility
behind most so-called ``nonreversible" MCMC algorithms and
\textbf{(b)} a method generalizing the MH rule to obtain
$(\pi,\involk)$-reversible kernels from arbitrary proposal kernels
$Q$. This is a generalisation of~\cite{tierney:1998} which aims to
provide a unifying and firm theoretical footing to recent and future
contributions. The framework encompasses, for example, both the
scenarios where $\pi$ and $Q$ have common dominating measure or when
$Q$ corresponds to a deterministic mapping.

\textbf{New challenges.} Novel applications have led to the development of highly sophisticated extensions of this basic scheme, prompted in particular by recent developments in the context of probability density
representation with normalising flows~\cite{baptista:2020,prangle:2019,papamakarios2019normalizing}, invertible neural networks ~\cite{ardizzone:2019}. For example, following the realisation that the textbook MH can be generalised by combining deterministic invertible mappings of the current state and a source of randomness in the proposal stage, some authors  have proposed using complex mappings involving both non-linearities and the composition of multiple layers~\cite{albergo:2019, thin:2020, spanbauer:freer:2020}, while~\cite{sherlock:2019,gustafson:1998} explore the use of nonreversible Markov kernels.  However it is not always clear that the resulting algorithms are convergent. In particular application of Markov chain theory may seem difficult at first sight given the new levels of complexity involved. Our aim in this paper is to provide users with simple to use theoretical guarantees ensuring validity of the algorithms.

\textbf{Contribution \#2: easy ready-made convergence results.}
Proving convergence of MH methods can be delicate in general. However,
in the $\pi$-reversible case, \cite{mengersen:tweedie:1996} and
\cite{tierney:1994} have derived simple conditions ensuring
convergence of $P$ in the case where $\pi$
and $Q$ share a common dominating measure $\mu$, for example the
Lebesgue measure when $\Zset=\rset^d$.

\begin{theorem}[Convergence of textbook MH] \label{theo:basic-cv-textbook}
  Assume that $\pi$ is not a Dirac mass function and has common $\sigma-$finite dominating measure $\mu$ with $\{Q(z,\cdot), z \in \Zset\}$. Denote $\pi$ and $\{q(z,\cdot), z \in \Zset\}$ the resulting densities. Suppose in addition that $\pi$ is not a Dirac mass and $Q(z,\msz^+) =1$ for any $z \not \in \Zset^+$ with $\Zset^+ = \{z \in \msz \, \colon\, \pi(z) >0\}$. If for any $z' \in \Zset$ such that $\pi(z')>0$ we have $q(z,z')>0$  for any $z \in \msz$, then 
  for any $f\colon\Zset\rightarrow\mathbb{R}$ such that $\pi(|f|)<\infty$, almost surely it holds that 
\begin{equation} \label{eq:averages-convergence}
\lim_{n\rightarrow \infty} n^{-1} \sum_{i=1}^n f(Z_i) = \pi(f) \eqsp.
\end{equation}
In addition, for all $z\in\Zset$
\begin{equation} \label{eq:TVconvergence}
\lim_{n\rightarrow \infty}\|P^n(z,\cdot) - \pi(\cdot)\|_{\rm TV} =0 \eqsp .
\end{equation}
\end{theorem}
Our second contribution is to provide similarly simple and easy to use conditions to establish convergence for $(\pi,\involk)$-reversible kernel. We translate these conditions to cover  complex proposal mechanisms based on conditional invertible neural transform ensuring that basic convergence properties hold in these novel settings. As we shall see some of the results lead to simple implementation suggestions ensuring that conclusions similar to those of \Cref{theo:basic-cv-textbook} hold. Establishing these properties is often overlooked and a necessary prerequisite to any more refined analysis characterising their performance, such as quantitative finite time convergence bounds.

\textbf{Contribution \#3: application to particular MCMC algorithms.}  We show how our conditions and construction can be used in practice to design $(\pi,\involk)$-reversible kernels which come with convergence guarantees. We first work out a generalization of the Hamiltonian Monte Carlo algorithm in which the gradients of the log-density in the leap frog steps are replaced by general neural transforms~~\cite{neal:2011,sohl2014hamiltonian}. Next, we derive and analyse two lifted Markov kernels~\cite{diaconis:holmes:neal:2000,chen:lovasz:pak:1999,turitsyn:2011,neklyudov:welling:vetrov:2020} covering obtained using conditional invertible transforms on an augmented space. Our experimental results (postponed to Supplementary paper) show numerically the benefits of nonreversibility in several sampling experiments.

The proofs of the main results and some facts, followed by a \suppl ,
can be found in the supplementary material and, for example,
(S{\color{red}\#}) refers to the {\color{red}\#}-th equation in the
supplement. The standard notation and definitions used are precisely
described in the supplementary \Cref{sec:notat-defin} for the reader's convenience.%  and we

\section{$(\pi, S)$-reversibility and the Generalized MH rule}\label{sec:mu-s-reversibility}

There has recently been renewed interest in the design of $\pi$-invariant Markov kernels which are nonreversible. In many scenarios, departing from reversibility can both improve the mixing time and
reduce the asymptotic variance of resulting estimators. It has been shown in~\cite{andrieu:livingstone:2019} that these nonreversible Markov kernels
fall under the same common framework of $(\pi,\involk)$-reversibility (introduced below) which encompasses
the modified (or skew) detailed balance conditions. Before proceeding further, additional notations are needed.  Let
  $\invol$ be an involution on $\Zset$, $\invol \circ \invol = \Id$ and
  $\involk$ be the associated kernel $\involk(z,A)= \indi{A}\bigl(\invol(z)\bigr)$,
  $z \in \Zset$, $A \in \Zsigma$.  Let $\dmu$ be a $\sigma$-finite
  measure on the product space $(\Zset^2, \Zsigma^{\otimes
    2})$  (the diacritic~$\check{}$~is used to denote measures on the product space $\Zset^2$). Denote by $\dmu^\invol = \pushf{(\sym_{\invol})}{\dmu}$ the
  pushforward of $\dmu$ by the transform
  $\sym_{\invol}(z,z') = \bigl(\invol(z'),\invol(z)\bigr)$:
%{\small{
\begin{equation}
\label{eq:s-symmetrization}
\dmu^\invol(C)= \int \indi{C}\bigl(\invol(z'),\invol(z)\bigr) \dmu\bigl(\rmd(z,z')\bigr) \eqsp,   C \in \Zsigma^{\otimes 2} \eqsp.
\end{equation}
%}}
%
Note that $\sym_{\invol}$ is an involution $\sym_{\invol} \circ \sym_{\invol}= \Id$ which implies that $(\dmu^\invol)^\invol= \dmu$.
\begin{definition}[after \cite{andrieu:livingstone:2019}]
\label{def:pi-S-reversible}
The measure $\dmu$ is $\invol$-symmetric if $\dmu = \dmu^\invol$.
The sub-Markovian kernel $P$ is $(\pi,\involk)$-reversible if the measure $\dmu_P$ defined as $\dmu_P\bigl(\rmd(z,z')\bigr) = \pi(\rmd z) P(z,\rmd z')$
% \begin{equation}
%   \label{eq:def:dmu_P}
% \dmu_P(d(z,z')) = \pi(\rmd z) P(z,\rmd z')
% \end{equation}
is $\invol$-symmetric.
\end{definition}
It is established in \Cref{sec:proof:condition-reversibility} that  $P$ is $(\pi,\involk)$-reversible if it satisfies the \textbf{skew detailed balance} condition,
\begin{equation}
\label{eq:condition-reversibility}
\pi(\rmd z) P(z,\rmd z')= \pushf{s}{\pi}(\rmd z') \involk P \involk(z', \rmd z) \eqsp.
\end{equation}
In particular, if $\pushf{\invol}{\pi} = \pi$ and $P$ is a Markov kernel,
then $\pi$ is invariant for $P$. We assume that the condition
$\pushf{\invol}{\pi} = \pi$ is in force in the rest of the paper. Note that for $\invol=\Id$ we recover the standard detailed balance condition (see \Cref{sec:standard-reversible-MH}).

\subsection{{Generalized Metropolis-Hastings}}
\label{sec:generalized_MH}
The MH algorithm gives a method to transform any proposal Markov kernel $Q$ into a $\pi$-reversible Markov kernel. We derive a Generalized
Metropolis-Hastings (GMH) rule to turn $Q$  into a   $(\pi,\involk)$-reversible Markov kernel. We then apply this condition to the case where $\pi(\rmd z')$ and $Q(z,\rmd z')$ have a density \wrt\ to a common dominating measure, and to the case where $Q(z,\rmd z')= \updelta_{\Phi(z)}(\rmd z')$ for $\Phi \colon \Zset \rightarrow \Zset$. We first establish a simple necessary and sufficient condition on the proposal kernel $Q$ and the acceptance probability function $\alpha \colon \msz^2\rightarrow\ccint{0,1}$  for the resulting (sub-Markovian) kernel
\begin{equation} \label{eq:Q_alpha}
  Q_{\alpha}(z,\rmd z') := \alpha(z,z') Q(z,\rmd z') \eqsp,
\end{equation}
to be  $(\pi, \involk)$-reversible.  A subset $A \subset \Zset^2$ is said to be $\invol$-symmetric if $(z,z') \in A$ \iff\ $\bigl(\invol(z'),\invol(z)\bigr) \in A$.  We denote
%{\small
\begin{equation}
  \label{eq:def:dmu_Q}
  \dnu\bigl(\rmd(z,z')\bigr) := \pi(\rmd z) Q(z,\rmd z') \eqsp.
\end{equation}
%}
The following result provides us with a key instrument to work with the densities of $\dnu$ and its pushforward $\dnu^\invol$ in full generality.
\begin{proposition}
  \label{prop:extension-tierney}
  Set $\dlambda= \dnu + \dnu^\invol$, $h= \rmd \dnu / \rmd \dlambda$ and
  $A_{\dnu} = \defEns{h \times h \circ \sym_{\invol} > 0} \in\Zsigma^{\otimes 2}$.
 Then, the restrictions $\dnu_{A}(\cdot)= \dnu( \cdot \cap A_{\dnu})$ and $\dnu^\invol_{A}(\cdot)= \dnu^s(\cdot \cap A_{\dnu})$ are equivalent and $\dnu_{A,\comp}(\cdot)= \dnu( \cdot \cap A_{\dnu}^{\comp})$ and $\dnu^\invol_{A,\comp}(\cdot)= \dnu^s(\cdot \cap A_{\dnu}^{\comp})$ are mutually singular. %The set $A_{\dnu}$ is unique up to sets which are negligible for both $\dnu$ and $\dnu^{\invol}$.
  In addition, define, for $(z,z') \in A_{\dnu}$,
$r(z,z')= h(z,z')/ h\bigl(\invol(z'),\invol(z)\bigr)$. Then, $r$ is a version of the density of $\dnu_{A}$ \wrt\ $\dnu_{A}^\invol$, \ie\  $r= \rmd \dnu_{A} / \rmd \dnu^\invol_{A}$ and $r(z,z')= 1 / r \circ \sym_{\invol}(z,z')$ for all $(z,z') \in A_{\dnu}$.
\end{proposition}
The following result applies Proposition~\ref{prop:extension-tierney} and extends the seminal result~\cite[Theorem~2]{tierney:1998} to the $(\pi,\involk)$-reversible case.
\begin{theorem}
\label{theo:extension-tierney}
The sub-Markovian kernel $Q_\alpha$ in \eqref{eq:Q_alpha} is $(\pi,\involk)$-reversible if and only if the following conditions hold.
\begin{enumerate}[wide, labelwidth=!, labelindent=0pt,label=(\roman*),noitemsep,nolistsep]
\item \label{item:extension-tierney-i} The function $\alpha$ is zero $\dnu$-\mae on $A^{\comp}_{\dnu}$.
%, where $\dnu$ is defined by \eqref{eq:def:dmu_Q} and $A_{\dnu}$ in \eqref{eq:def_A_dnu}.
\item \label{item:extension-tierney-ii} The function $\alpha$ satisfies $\alpha(z,z') r(z,z')= \alpha(\invol(z'),\invol(z))$ $\dnu$-\mae on $A_{\dnu}$.
\end{enumerate}
\end{theorem}
Similarly to the $\pi$-reversible case, we define the generalized Metropolis-Hastings (GMH)
rejection probability by 
%{\small{
\begin{equation}
\label{eq:generalized-MH}
  \alpha(z,z') =
  \begin{cases}
    \accfun\left(\frac{h(\invol(z'),\invol(z))}{h(z,z')}\right) & h(z,z') \ne 0, \\
    1 & h(z,z')= 0,
  \end{cases}
\end{equation}
%}
%}

\vspace{-0.1cm}
\noindent
where $\accfun\colon \rset_+^* \to \ccint{0,1}$ satisfies $\accfun(0)= 0$ and for $t \in \rset_+^*$,
\begin{equation}
  \label{eq:prop_accfun}
    t\accfun(1/t) = \accfun(t) \eqsp.
  \end{equation}
  Then $\alpha$ satisfies the conditions~\ref{item:extension-tierney-i}-\ref{item:extension-tierney-ii} of \Cref{theo:extension-tierney}, see \Cref{sec:proof:generalized-MH}.
We may take for example  $\accfun(t)= \min(1,t)$ or $\accfun(t)= t/(1+t)$
which correspond to the classical Metropolis-Hastings and Barker ratio respectively.

We can obtain the GMH Markov kernel $P$ which is $(\pi,\involk)$-reversible by adding Dirac masses:
{\small{\begin{equation}
     \label{eq:kernel-with-mass}
     P(z,\rmd z') = Q_{\alpha}(z,\rmd z')+ a(z) \updelta_{z}(\rmd z') + b(z) \updelta_{\invol(z)}(\rmd z')  \end{equation}
 }}
 
 \vspace{-0.7cm} 
 \noindent
with $a$, $b$ nonnegative, measurable satisfying $a(z)= a(\invol(z))$ and $a(z)+b(z)= 1-Q_{\alpha}(z,\msz)$; see \Cref{sec:proof:expressions-a-b}. In the sequel, we focus on the case $a(z)=0$ and $b(z)= 1 - Q_\alpha(z,\msz)$.

\subsection{{GMH for particular proposal maps}}
\label{sec:proposals}
We now specialize~\eqref{eq:generalized-MH} to the case where $\pi$ and $Q$ admit a common dominating measure and the case where $Q$ is deterministic.

\textbf{Proposal with densities}.
Suppose there is a common dominating measure $\mu$ on $(\Zset,\Zsigma)$ such that $\pi(\rmd z) = \pi(z) \mu(\rmd z)$, $Q(z,\rmd z')= q(z,z') \mu(\rmd z')$ and that $\mu$ is invariant by $\invol$, \ie\ $\pushf{\invol}{\mu}= \mu$. In this scenario, we have (see \Cref{sec:proof:proposal-with-densities})
% $h = \tilde{h}/\{\tilde{h} + \tilde{h} \circ \sym_{\invol}\}$, $\tilde{h}(z,y) = \pi(z) q(z,z')$, $A_{\dnu}$ is defined by \eqref{eq:def_A_dnu},
%{\small
\begin{equation}
  \label{eq:def_A_dnu}
A_{\dnu} = \defEns{\pi(z) q(z,z') \times \pi(z') q\bigl(\invol(z'),\invol(z)\bigr) > 0} \eqsp.
\end{equation}
%}
In addition, we obtain using~\eqref{eq:generalized-MH} that
{\small{
\begin{equation}
  \label{eq:alpha_density}
\alpha(z,z')= \begin{cases} \accfun\parentheseDeux{\frac{\pi(z') q(\invol(z'),\invol(z))}{\pi(z) q(z,z')}} & \pi(z) q(z,z') \ne 0, \\
  1  & \pi(z) q(z,z')=0 \eqsp.
\end{cases}
\end{equation}
}}

\vspace{-0.5cm}
\noindent
Theorem \ref{theo:basic-cv-textbook} exploits the fact that in the $\pi$-reversible scenario the MH kernel is $\pi$-irreducible if the condition $\pi(z') > 0$ implies that $q(z,z') > 0$~\cite{mengersen:tweedie:1996}. This result can be extended to the $(\pi,\involk)$-reversible case as follows.
\begin{lemma}
The GMH Markov kernel $P$ in \eqref{eq:kernel-with-mass} is $\pi$-irreducible if, $\pi(z') > 0$ implies that, for all $z \in \Zset$, $q(z,z') > 0$ and $q\bigl(\invol(z), \invol(z')\bigr) > 0$.
\end{lemma}
In the $\pi$-reversible case, \cite[Corollary~2]{tierney:1994} shows that the $\pi$-irreducibility condition implies that the GMH Markov kernel $P$ \eqref{eq:kernel-with-mass} is Harris recurrent and aperiodic.
These two properties have consequences that are very important in practice: the convergence in total variation of the iterates of the kernel to the invariant distribution and the ergodic theorem become valid also for all the initial conditions. These results extend to $(\pi,\involk)$-Markov kernels (see \Cref{sec:proof:theo:irred_and_co}).
\begin{theorem}\label{theo:irred_and_co}
  Let $P$ be defined as in~\eqref{eq:kernel-with-mass}, with $a(z)=0$ and $b(z)= 1 - Q_\alpha(z,\msz)$. Assume that for any $z' \in\msz$, $\pi(z') > 0$ implies $q(z,z') \times q\bigl(\invol(z),\invol(z')\bigr) >0$
  for any $z \in \msz$. Suppose in addition that $\pi$ is not a Dirac mass and $Q(z,\msz^+) =1$ for any $z \not \in \Zset^+$ with $\Zset^+ = \{z \in \msz \, \colon\, \pi(z) >0\}$. The conclusions of \Cref{theo:basic-cv-textbook}  hold.
%Suppose in addition that $\pi$ is not a Dirac mass and $Q(z,\msz^+) =1$ for any $z \not \in \Zset^+$ with $\Zset^+ = \{z \in \msz \, \colon\, \pi(z) >0\}$.  Then, the conclusions of Theorem \ref{thm:tierney94} hold with~\eqref{eq:averages-convergence} and~\eqref{eq:TVconvergence} valid for all $z \in \Zset$.
%\begin{enumerate}[label=(\arabic*)]
%\item \label{theo:irred_and_co_i}
% For any $z \in \Zset$, $\tvnorm{P^n(z,\cdot) - \pi}= 0$.
%\item  \label{theo:irred_and_co_ii} In addition, for any
%  $f : \msz \to \rset$, satisfying $\pi(\abs{f}) < \plusinfty$, almost
%  surely
%  $ \lim_{n \to\plusinfty} [n^{-1} \sum_{k=0}^{n-1} f(Z_k) ] =
%  \pi(f)$, where $(Z_k)_{k \in\nset}$ is a Markov chain associated
%  with $P$ and starting from $z_0$.
%\end{enumerate}
\end{theorem}

\paragraph{Deterministic proposal.}
Suppose now that $\transfdet$ is a one-to-one mapping from $\Zset$ onto $\Zset$ such that
\begin{equation}
\label{eq:extended-involution}
\transfdet^{-1}= \invol \circ \transfdet \circ \invol \eqsp.
\end{equation}
We consider the deterministic proposal kernel $Q(z,\rmd z')= \updelta_{\transfdet(z)}(\rmd z')$: when the current state is $z$, the proposal is $\transfdet(z)$.
Condition~\eqref{eq:extended-involution} implies that $F= \invol \circ \transfdet$ is an involution.
Our setting covers involutive MCMC -- corresponding to the case $\invol = \Id$ introduced in~\cite[Section~2]{tierney:1998} and more recently in~\cite{neklyudov:welling:vetrov:2020}.

In this scenario we have that (see \Cref{sec:proof:deterministic-case})
$\dnu\bigl(\rmd(z,z')\bigr)= \pi(\rmd z) \updelta_{\transfdet(z)}(\rmd z')$ and $\dnu^\invol\bigl(\rmd(z,z')\bigr)= \pi(\rmd z') \updelta_{\transfdet^{-1}(z')}(\rmd z)$. The function $h$ defined in \Cref{prop:extension-tierney} is given by  $h(z,z') = \indi{\transfdet(z)}(z') k(z)$ with
\begin{equation}
\label{eq:definition-k-lambda}
k(z) = \frac{\rmd \pi}{\rmd \lambda}(z), \quad \lambda= \pi + \pushf{(\transfdet^{-1})}{\pi} \eqsp.
\end{equation}
\Cref{theo:extension-tierney} is satisfied with the acceptance probability given by $\alpha\bigl(z,\transfdet(z)\bigr)= \bar{\alpha}(z)$ with
\begin{equation}
\label{eq:generalized-MH-deterministic}
\bar{\alpha}(z)=  \accfun\left(\frac{k\bigl(\invol \circ \transfdet (z)\bigr)}{k(z)} \right)
\end{equation}
if $k(z) > 0$ and $\bar{\alpha}(z)=1$, otherwise. Of course, there is no need to define $\alpha(z,z')$ for $z' \ne \transfdet(z)$. A special case of interest is when $\Zset= \rset^d$ and the target distribution $\pi(\rmd z)= \pi(z) \Leb_d(\rmd z)$ has a density \wrt\ the Lebesgue measure on $\rset^d$. Here the dominating measure $\lambda$ is given by
\begin{equation}
\label{eq:definition-lambda-dens}
\lambda(\rmd z)=\{ \pi(z) + \pi \circ \transfdet(z) \Jac_{\transfdet}(z) \} \Leb_d(\rmd z) \eqsp,
\end{equation}
where $\Jac_f$ denotes the absolute value of the Jacobian determinant of $f$.
Then, the density $k(z)$ is given by
\begin{equation}
\label{eq:definition-k-lambda-dens}
k(z)= \frac{\pi(z)}{\pi(z) + \pi \circ \transfdet(z) \Jac_{\transfdet}(z)}
\end{equation}
and the acceptance ratio $\bar{\alpha}(z)$ takes the simple form
\begin{equation}
\label{eq:generalized-MH-deterministic-dens}
\bar{\alpha}(z)=  \accfun\left(\frac{\pi \circ \transfdet(z) \Jac_{\transfdet}(z)}{\pi(z)} \right)
\end{equation}
if $\pi(z) \ne 0$ and $\bar{\alpha}(z)=1$ otherwise (see \Cref{sec:proof:generalized-MH-deterministic-dens}). We obtain the same acceptance ratio given by~\cite[Eq.~(5)]{neklyudov:welling:vetrov:2020}, which derive this expression in the special case $\invol= \Id$ and thus $\transfdet^{-1}= \transfdet$ is an involution. It is perhaps striking that the acceptance ratio \textbf{does not depend} on $\invol$: this comes from the fact the target distribution $\pi$ is invariant by $\invol$.
This setting encompasses many algorithms, HMC~\cite{duane:1987,neal:2011}, and NICE-MC~\cite{song2017nice} -- see below, \cite{neklyudov:welling:vetrov:2020} and the references therein.
Of course, in most cases, $(\pi,\involk)$-reversible deterministic Markov kernels are not $\pi$-irreducible and Harris recurrent. They can nevertheless be important building blocks of Markov kernels as in the HMC construction.

%%% Local Variables:
%%% mode: latex
%%% TeX-master: "main-t"
%%% End:

\section{Applications and examples}
\label{sec:applications-examples}
% !TEX root = main-t.tex

\subsection{Generalized Hamiltonian Dynamics}
\label{sec:gener-hamilt-dynam}
We first consider generalizations of the Hamiltonian Monte Carlo algorithm (see~\cite{neal:2011,sohl2014hamiltonian}). These methods might also be seen as a special case of NICE (Non-linear Independent Components Estimation) MCMC methods~\cite{song2017nice,neklyudov:welling:vetrov:2020}. The objective is to sample a distribution on $\rset^d$ of density $\pi_0$ \wrt~the Lebesgue measure.
% $\Leb_d$.
We use a data augmentation approach which consists of adding a ``momentum'' variable with stationary distribution admitting a symmetric density $\varphi$ on $\rset^d$ \wrt~the Lebesgue measure, \eg~$\varphi(-p) = \varphi(p)$. More precisely, on the extended state space $\Zset = \rset^{2d}$, we consider the extended target density defined by $\pi(x,p)= \pi_0(x) \varphi(p)$ and the Markov chain $\big(Z_i=(X_i,P_i)\big)_{i \in \nset}$. The involution is taken to be $\invol(x,p)= (x,-p)$. By construction, $\pushf{\invol}{\pi} = \pi$.
%It is easily seen that $\pi\bigl(\invol(x,p)\bigr)= \pi(x,p)$ and $\pushf{s} \Leb_{2d} = \Leb_{2d}$.

We first show how to construct a  $(\pi,\involk)$-reversible Markov kernel on $\rset^{2d}$ using modified leap-frog integrators.  Let $m \in \nset$ and $\{\transfm_i,\transfn_i\}_{i=1}^m$ be $\rmc^1$ functions on $ \rset^{d}$. We define a mapping $\transfdet(x,p) = \LF_{m} \circ \cdots \circ \LF_1(x,p)$ on $\rset^{2d}$ where $\LF_i$ is given by $(x_{i+1},p_{i+1}) = \LF_{i}(x_i,p_i)$ where for $h>0$
%{\small{
\begin{equation}
  \label{eq:definition-nice-flow}
  \begin{cases}
    p_{i+1/2} &= p_{i} + h \transfm_i(x_i) \eqsp, \\
    x_{i+1} &= x_i +h p_{i+1/2} \eqsp, \\
    p_{i+1} &= p_{i+1/2} + h\transfn_i(x_{i+1}) \eqsp.
  \end{cases}
\end{equation}
%}
%}
%
%\vspace{-0.5cm}
%\noindent
It is easily seen that $\LF_{i}$ is a $\rmc^1$ diffeomorphism on $\rset^{2d}$ to $\rset^{2d}$ with $\Jac_{\LF_i}(x,p)= 1$. Moreover, if for any $i\in\{1,\dots,m\}$, $\transfm_i = \transfn_{m+1-i}$, then $\invol\circ\transfdet\circ\invol=\transfdet^{-1}$; see \Cref{subsec:NICE_proofs}. We assume in the sequel that this condition holds. 
Consider now the Markov kernel 
{\small{
\begin{align}
\label{eq:NICE_kernel}
&P((x,p), \rmd (y,q)) = \bar{\alpha}(x,p) \updelta_{\transfdet(x,p)}(\rmd (y,q)) \\
& \qquad \qquad + (1-\bar{\alpha}(x,p))\updelta_{(x,-p)}(\rmd (y,q))\eqsp,\\
\label{eq:generalized-MH-deterministic-nice}
&\text{with}\qquad  \bar{\alpha}(x,p)= \accfun\left( \pi \circ \transfdet(x,p)/\pi(x,p) \right)\eqsp,
\end{align}}}
\noindent 
if $\pi(x,p) >0$ and $\bar{\alpha}$ is equal to 1 otherwise. Using \eqref{eq:generalized-MH-deterministic-dens}, $P$ is $(\pi, \involk)$-reversible,
%Following the method outlined in \Cref{sec:mu-s-reversibility}, we first identify the dominating measure $\lambda$ (see \eqref{eq:definition-k-lambda}  which is easily shown to have a density \wrt\ $\Leb_{2d}$ given by \suppl\ :
%\begin{equation}
%\label{eq:definition-lambda}
%\lambda(\rmd (x,p)) = \{ \pi (x,p) + \pi \circ \transfdet(x,p) \} \Leb_{2d}(\rmd (x,p))
%\end{equation}
%The density $k$ (see \eqref{eq:definition-k-lambda}) is respectively defined as
%\begin{equation}
%\label{eq:definition-density-nice}
%k(x,p)= \frac{\rmd \pi}{\rmd \lambda}(x,p)= \frac{\pi(x,p)}{\pi (x,p) + \pi \circ \transfdet(x,p)} \eqsp.
%\end{equation}
%Since $\pi \circ \invol= \pi$ and $\transfdet \circ \invol \circ \transfdet= \invol$, we  obtain similarly
%\begin{equation}
%\label{eq:definition-density-nice-transformed}
%k( \invol \circ \transfdet(x,p))= \frac{\pi(\transfdet(x,p))}{\pi(x,p) + \pi(\transfdet(x,p))}
%\end{equation}
but is deterministic and therefore not
ergodic. A standard approach to address this issue, used in the context of HMC algorithms, is to refresh the
momentum between two successive moves according to a Markov transition
preserving the distribution $\varphi$. A particular choice consists of sampling the velocity afresh from $\varphi$ before applying the kernel~\eqref{eq:NICE_kernel}. More precisely, we define the Markov chain $(X_i)_{i\in \nset}$ by the following recursion. From a state $X_{k}$, the $k+1$-th iterate is defined by: \begin{enumerate*}
    \item sample $P_{k+1}$ from $\varphi$ and set $(Y_{k+1},Q_{k+1}) = \transfdet(X_{k},P_{k+1})$; accept $X_{k+1} = Y_{k+1}$ with probability $\bar{\alpha}(X_k,P_{k+1})$ and reject $X_{k+1}=X_k$ otherwise. 
\end{enumerate*}
In this case one can check that $(X_i)_{i \in \nset}$ is a Markov chain on $\rset^d$ of kernel, obtained by marginalisation of~\eqref{eq:kernel-with-mass} \wrt\ the momentum distribution,
 {\small\begin{equation}
\label{eq:marginal-NICE-kernel}
K(x,\rmd y)= K_\alpha(x,\rmd y) + \{1- \bar{\alpha}(x) \} \updelta_{x}(\rmd y) \eqsp,
\end{equation}
}
where $\bar{\alpha}(x) = K_{\alpha}(x,\rset^d)$ and denoting 
$G_x(p)= \projq \circ \transfdet(x,p)$, $\projq(x,p)= x$,
\begin{equation}
\label{eq:definition-K-alpha-1}
K_\alpha(x,\rmd y)= \int \bar{\alpha}(x,p) \varphi(p) \updelta_{G_x(p)}(\rmd y) \rmd p
\end{equation} 
%From a practical point of view, note that simulating from $K_{\alpha}$ just consists in 
% Such kernels are deterministic and therefore (most likely) not
% ergodic. Similarly to HMC, a possible solution is to refresh the
% momentum between two successive NICE moves according to a Markov transition
% preserving the density $\varphi$. For simplicity, we consider in this
% paper only a ``full refresh", which amounts to draw a new value of the
% momentum at each iteration from $\varphi$.  Combining the refresh step
% with the NICE move and marginalizing \eqref{eq:kernel-with-mass} \wrt\ the momentum random
% variable defines a Markov kernel on $\rset^d$
%  \begin{equation}
% \label{eq:marginal-NICE-kernel}
% K(x,\rmd y)= K_\alpha(x,\rmd y) + \{1- K_{\alpha}(x,\rset^d) \} \updelta_{x}(\rmd y) \eqsp,
% \end{equation}
% where  $K_\alpha(x,\rmd y)= \int \bar{\alpha}(x,p) \varphi(p) \updelta_{\transfdet_1(x,p)}(\rmd y) \rmd p$ with
% $\transfdet_1(x,p)= \projq \circ \transfdet(x,p)$, $\projq(x,p)= x$.
If for any $x \in \rset^d$, $p \mapsto G_x(p)$ is a diffeomorphism on $\rset^d$, then \Cref{theo:irred_and_co} can be applied. In such case, $K_\alpha(x,\rmd y)= \alpha(x,y) q(x,y)$
with 
\begin{align}
\label{eq:acceptance-NICE}
& \alpha(x,y) = \accfun\left(\frac{\pi_0(y) \varphi\bigl\{H_x\bigl(G_x^{-1}(y)\bigr)\bigr\}}{\pi_0(x) \varphi\bigl(G_x^{-1}(y)\bigr)} \right) \,, \\
\label{eq:definition-q-nice}
&q(x,y)= \varphi\bigl(G_x^{-1}(y)\bigr) \Jac_{G_x^{-1}}(y) \,,
\end{align}
and $H_x(p)= \projp \circ \transfdet(x,p)$ and $\projp(x,p)=p$. 
The expression of $\alpha(x,y)$ is only of theoretical interest and  is not needed to implement the algorithm.  
Of course, requiring that $G_x$ is a diffeomorphism imposes conditions on $\LF_i$, $i \in  \{1,\dots, m\}$ and \Cref{theo:lip_nice} (see \Cref{sec:proof:acceptance-NICE}).

\begin{theorem}
\label{theo:condition-diffeo}
Assume that $\varphi>0$, $\varphi(-p)=\varphi(p)$ for all $p\in\rset^d$ and for any $i \in  \{1,\dots, m\}$, $M_i$ and $N_i$ are $\mathrm{L}$-Lipschitz and  $h \leq c_0/[\mathrm{L}^{1 / 2}m]$, where $c_0 \approx 0.3$ (see \Cref{theo:lip_nice}). Then for any $x \in \rset^d$, $p \mapsto G_x(p)$ is a $\rmC^1$-diffeomorphism. 
%In addition, the conclusions of \Cref{theo:irred_and_co} hold for $K$ defined by \eqref{eq:marginal-NICE-kernel}.
\end{theorem}

The proof of this result is along the same lines as the proof of
\cite[Theorem 1]{durmus:moulines:saksman:2017} which focuses on the standard HMC algorithm.

A by-product of the proof of \Cref{theo:condition-diffeo}, is that, perhaps surprisingly (see \eqref{sec:proof:NICE-identity-ratio}) 
\begin{equation}
\label{eq:NICE-identity-ratio}
{q(y,x)}/{q(x,y)}= {\varphi\bigl(H_x \circ G_x^{-1}(y)\bigr)}/{\varphi\bigl(G_x^{-1}(y)\bigr)} \eqsp,
\end{equation}
implying that $\alpha$ \eqref{eq:acceptance-NICE} is the textbook MH acceptance ratio corresponding to $q$ in \eqref{eq:definition-q-nice}, and  the Markov kernel $K$ \eqref{eq:marginal-NICE-kernel} is therefore $\pi$-reversible. It easily checked that this kernel satisfies the conditions of \Cref{theo:basic-cv-textbook} and the convergence results apply.
%Therefore, we have established the following result. 
%\begin{theorem}
%\label{reversibility_nice}
%Under the assumptions of \Cref{theo:condition-diffeo}, $K_{\alpha}$ is the standard $\pi$-reversible MH kernel associated with the proposal density \eqref{eq:definition-q-nice}.
%\end{theorem}

%By \Cref{reversibility_nice}, $(X_i)_{i \in \nset}$ is time-reversible, which 
The $\pi_0$-reversibility of $K$ has the disadvantage of loosing the potentially advantageous non-backtracking (or persistency) features of $P$. It is possible to recover persistency by considering the mixture of kernels on the extended space $\rset^{2d}$ $\omega P + (1-\omega) L$ where $P$ is the deterministic kernel \eqref{eq:NICE_kernel} and
$L((x,p),\rmd (y,q)) = K(x, \rmd y) \varphi(q) \rmd q$. In words, we refresh independently the position and the momentum.  
The amount of persistency is controlled by $\omega$. 
%$ K \otimes (\varphi \times \Leb) + (1-\omega) P =:L\colon \Zset \times \Zsigma \rightarrow [0,1]$ for $\omega \in [0,1]$, which leads to a nonreversible Markov chain $\big(Z_i=(X_i,P_i)\big)_{i \in \nset}$. 
Theorem \ref{theo:condition-diffeo} establishes $\pi_0-$irreducibility of $K$ (see its proof), which immediately implies $\pi$-irreducibility of $L$;
see \Cref{subsubsec:onpersistency} for a more detailed discussion. 
%We stress that this notion of persistency is 

%%% Local Variables:
%%% mode: latex
%%% TeX-master: "main-t"
%%% End:

% !TEX root = main-t.tex

\subsection{Lifted kernels}
In this section, we apply the results of \Cref{sec:mu-s-reversibility}
to lifted kernels introduced in~\cite{diaconis:holmes:neal:2000,chen:lovasz:pak:1999,turitsyn:2011,michel:2016}. As
above, let $\target$ be a target probability density on $\rset^d$
\wrt~the Lebesgue measure. We extend the state space with a direction,
\ie\ we consider $\msz = \rset^d \times \msv$ with $\msv = \{-1,1\}$
and the extended target distribution
$\exttarget = \target \otimes [\{\updelta_{-1} +
\updelta_{1}\}/2]$. In this scenario the involution is
$\invf(x,v) = (x,-v)$.

\paragraph{Proposal with densities.} Let $q_{-1}(x,\cdot),q_{1}(x,\cdot)$  be two
transition densities \wrt\ the Lebesgue
measure on $\rset^d$. Consider a proposal kernel $Q((x,v), \rmd(y,w))$ with density $q((x,v), (y,w)) $ with
respect to $\Leb_d(\rmd y) \otimes \{\updelta_{-1}(\rmd w) + \updelta_{1}(\rmd w)\}$ given by
{\small\begin{equation}
\label{eq:lifted_density}
q\bigl((x,v), (y, w)\bigr)= \bigl\{\probav \indi{v}(w) +(1-\probav) \indi{-v}(w)\bigr\}  q_w(x,y),
\end{equation}}

\vspace{-0.7cm}
\noindent
where $\probav \in \ooint{0,1}$. In words, starting from $(x,v)$, we either ``keep" $w=v$ with probability $\probav$ or
``flip" $w=-v$ the direction otherwise, and then propose a
candidate $y$ according to $q_{w}(x,\cdot)$. In the original implementation of the lifting procedure \cite{turitsyn:2011}, $\probav$ is set to $1$; taking $\probav <1$ simply prevents the algorithm from getting ``stuck" in one direction which could impede convergence of the algorithm.

From~\eqref{eq:generalized-MH} and~\eqref{eq:alpha_density}, the
acceptance ratio $\alpha$ writes, for
$q_w(x,y)\target(x) \neq 0$, see \Cref{sec:proof:alpha_lifted},
% {\small\begin{multline}
% \label{eq:alpha_lifted}
% \alpha((x,v),(y,w))= \accfun\left(\frac{q_{-v}(y,x) \target(y)}{q_v(x,y)\target(x)}\right) \indi{v}(w) \\
% +\accfun\left(\frac{q_{v}(y,x) \target(y)}{q_{-v}(x,y)\target(x)}\right) \indi{v}(-w) \eqsp.
% \end{multline}}
%{\small
\begin{equation}
\label{eq:alpha_lifted}
\alpha\bigl((x,v),(y,w)\bigr)= \accfun\left(\frac{q_{-w}(y,x) \target(y)}{q_w(x,y)\target(x)}\right) \eqsp,
\end{equation}
%}
%
%\vspace{-0.5cm}
%\noindent
and $\alpha\bigl((x,v),(y,w)\bigr) = 1$ otherwise, where $\accfun$ satisfies
\eqref{eq:prop_accfun}. The GMH kernel is given by
\eqref{eq:kernel-with-mass} with $a(z)=0$ and
$b(z)= 1 - Q_{\alpha}\bigl((x,v),\Zset\bigr)$. Note that if the proposal move is
rejected, then the direction is automatically flipped.

In the case
$q_{-1} = q_{1}$, then the acceptance probability $\alpha$ \eqref{eq:alpha_lifted} does not
depend on $v,w$ and the GMH kernel~\eqref{eq:kernel-with-mass} can be
marginalized \wrt~$v$ yielding the $\pi_0$-reversible MH
algorithm of proposal density $q_1$. Since $\probav \in \ooint{0,1}$, the expression for $q$ in
\eqref{eq:lifted_density} implies the following result.
\begin{proposition}
  \label{propo:lifted_density}
Assume that  for any
$y \in \rset^d$, $\pi_0(y) >0$ implies $q_{-1}(x,y) >0$ and
$q_1(x,y) >0$, for all $x\in\rset^d$.  Then the conditions of
\Cref{theo:irred_and_co} hold, and the  GMH kernel~\eqref{eq:kernel-with-mass} is
ergodic.
\end{proposition}

Similarly to \Cref{sec:gener-hamilt-dynam}, the proposal densities $q_{v}(x,\cdot)$ are often associated to $\rmc^1$-diffeomorphisms $G_{v,x}\colon p\mapsto G_{v,x}(p)$. From a state $X_k$, we sample $P_{k+1}$ from $\varphi$ positive density on $\rset^{d}$ and set $Y_{k+1} = G_{V_k, X_k}(P_{k+1})$. 
%deterministic transforms $G_v\colon \rset^d\times \rset^d\to \rset^d$, for $v \in \msv$, satisfying  for any $x\in \rset^d, v\in \msv$, that $p\mapsto \tilde{G}_{v,x}(p) = G_v(x,p)$ is a $\rmc^1$-diffeomorphism, which can be established as in Subsection \ref{sec:gener-hamilt-dynam}. 
In this case,  
\begin{equation}
  \label{eq:def_q_v}
  q_v(x,y) = \varphi\bigl(G_{v,x}^{-1}(y)\bigr) \Jac_{G_{v,x}^{-1}}(y)  \eqsp.
\end{equation}
%using the change of variable $y = G_{v,x}(p)$.
%$Q_vf(x)=\int f\bigl(G_{v,x}(p)\bigr) \varphi(p) \rmd p$.

We illustrate the construction above with two examples of mappings $G$ satisfying the conditions we consider.
\begin{example}[(MALA-cIT) lifted kernel]
  Assume that $\pi_0$ is positive and continuously differentiable. For $x\in\rset^d$, we define two transforms $G_{1,x}, G_{-1,x}$. For $G_{1,x}$, we set
\begin{equation}
G_{1,x}\colon p\mapsto x + \gamma \nabla \log\pi(x) + \sqrt{2\gamma} p\eqsp,
\end{equation}
which corresponds to the proposal of the Metropolis Adjusted Langevin Algorithm (MALA).
In particular, for any $x\in\rset^d$, the transformation $G_{1,x}$ is a $\rmc^1$-diffeomorphism, with $\Jac_{G_{1,x}}(p) =  (2 \gamma)^{d/2}$ and 
\begin{equation}
G_{1,x}^{-1}(y) = \{y-x-\gamma \nabla \log \pi (x)\}/\sqrt{2 \gamma} \eqsp.
\end{equation}
\label{example:cIT}
For $G_{-1,x}$ we consider conditional invertible transforms~\cite{ardizzone:2019}
\begin{equation}
%  \label{eq:3}
  G_{-1,x}(p) = \LG_{K,x}\circ\dots\circ \LG_{1,x}(p) \eqsp,
\end{equation}
where for $i \in \{1,\ldots,K\}$, $\LG_{i,x}$ splits its input
into two parts $(p_{i,1},p_{i,2})\in\rset^{d_{i,1}}\times\rset^{d-d_{i,1}}$ and applies affine transformations between them
\begin{equation}
\begin{aligned}
p_{i+1,1}& = p_{i,1} \odot \exp\bigl(R_{i,1}(p_{i,2},x)\bigr) + M_{i,1}(p_{i,2},x) \eqsp,\\
p_{i+1,2}& = p_{i,2} \odot \exp\bigl(R_{i,2}(p_{i+1,1},x)\bigr) + M_{i,2}(p_{i+1,1},x) \eqsp.
\end{aligned}
\end{equation}
Here $R_{i,1},M_{i,1}$ (resp. $R_{i,2},M_{i,2}$) are any functions from $\rset^{d_{i,1}}$ (resp. $\rset^{d-d_{i,1}}$) to $\rset^d$.
This structure is an extension of the affine coupling block
architecture suggested in~\cite{dinh:2016}. Note that for any
$i \in \{1,\ldots,K\}$, $\LG_{i,x}$ is a $\rmC^1$-diffeomorphism on
$\rset^d$ of Jacobian determinant given by
$\Jac_{\LG_{i,x}}(p) = \exp\bigl(R_{i,1}(p_2,x)+R_{i,2}(p_1',x)\bigr)$. Therefore,
$G_{-1,x}$ is a
$\rmC^1$-diffeomorphism with Jacobian determinant which can be explicitly computed.~\eqref{eq:def_q_v} gives a nonreversible MH algorithm with convergence guarantees provided by \Cref{propo:lifted_density};
see details in \Cref{subsec:cIT_implementation}.

%Setting $G_1 = G_{\MALA}$ and $G_{-1} = G_{\cIT}$ in
%Thus,~\eqref{eq:def_q_v} gives a nonreversible MH algorithm with convergence guarantees provided by \Cref{propo:lifted_density}.
\end{example}

A specific case corresponds to the choice
\begin{equation}
  \label{eq:def_forward_back_G_transf}
  G_{v,x}(p) = \projq \circ \transflif^v(x,p) \eqsp,
\end{equation}
where $\transflif$
is a $\rmc^1$-diffeomorphism on $\rset^{2d}$.
We establish in the following result an
alternative expression for $\alpha$ using~\eqref{eq:alpha_lifted} and~\eqref{eq:def_q_v}, which relies on $\transflif^v$ and
$\Jac_{\transflif^v}$ and for which $\Jac_{G_{v,x}}$ is not required anymore (see \Cref{sec:proof:altern_accfun}).
\begin{lemma}
  \label{lem:altern_accfun}
  Assume that, for any $(x,v) \in \Zset$, the mapping $G_{v,x}$ is a
$\rmc^1$-diffeomorphism on $\rset^d$. Then, for any $x,y \in \rset^d$, $v,w \in \msv$, the acceptance ratio $\alpha$ defined in \eqref{eq:alpha_lifted}% corresponding to $q_v$ defined by
  % \eqref{eq:def_q_v}
  is given by
%{\small 
\begin{equation}
%    &\alpha(x,v;y,w)
    % = \phi\left(\frac{\Jac_{G^{-1}_{-v,y}}(x) \varphi(G^{-1}_{-v,y}(x)) \pi(y)}{\Jac_{G^{-1}_{v,x}}(y) \varphi(G^{-1}_{v,x}(y))\pi(x)}\right) \indi{v}(w)
   % \\
 \accfun\left(\frac{ \mu\Bigl(\transflif^{w}\bigl(x, G^{-1}_{w,x}(y)\bigr)\Bigr)}{\mu\bigl(x, G^{-1}_{w,x}(y)\bigr)}\Jac_{\transflif^w}\bigl(x, G^{-1}_{w,x}(y)\bigr)\right)  \eqsp,
\end{equation}%}
where $\mu(x,p) = \pi_0(x)\varphi(p)$.
\end{lemma}
This result is of practical interest because in many cases, the computation of $\Jac_{\transflif^v}(x,p)$ is much simpler than that of $\Jac_{G_{v,x}}(p)$. As an example, if $\transflif$ is the generalized HMC transform $\transflif= \LF_m\circ\dots\circ\LF_1$ where $\LF_i$ is defined in \eqref{eq:definition-nice-flow}, $\Jac_{\transflif^v}(x,p)=1$ while $\Jac_{G_{v,x}}(p)$ has no simple closed-form expression.
%% Here we can't write the diffeomorphism  $\tidle{G}_{v,x}$ ony for generalized MALA in practice in experiments -> write it in previous section ?
% \begin{example}
% Consider real-NVP \cite{dinh:krueger:bengio:2014} like transitions for our MCMC algorithm.
% This framework could include in particular the Generalized Hamiltonian Monte Carlo algorithm presented in \citep{levy:hoffman:sohl-dickstein:2017,gu:sun:2020}.
% Consider families of transformations $\{s_k\}_{k=1}^{K}$, $\{t_k\}_{k=1}^K$ and $\{M_k\}_{k=1}^K$, $\{N_k\}_{k=1}^K$.
% We write the transitions as follows $T_i (x_i, p_i) = (x_{i+1}, p_{i+1})$, where
% \begin{equation}
% \begin{cases}
% p_{i+1/2} &= \exp\{h s_i(x_i)\} \odot p_i + h M_i(x_i) \\
% x_{i+1} &= x_i + h p_{i+ 1/2} \\
% p_{i+1} &= \exp\{h t_i(x_{i+1})\} \odot p_{i+1/2} + h N_{i}(x_{i+1})
% \end{cases}.
% \end{equation}
% The inverse of this mapping is easily obtained \suppl.
% \end{example}
\paragraph{Deterministic proposals.}

Using a $\rmC^1$-diffeomorphism $\Psi$ on $\rset^{2d}$, we may also consider deterministic moves like in \Cref{sec:gener-hamilt-dynam}. Consider the extended state space
  $\msz = \rset^{2d} \times \msv$, the target distribution
  $\pi = \pi_0 \otimes \varphi \otimes
  [\{\updelta_{-1}+\updelta_{1}\}/2]$, where $\varphi$ is a symmetric
  density \wrt~$\Leb_d$, and the involution  $\invol(x,p,v) =   (x,p,-v)$.
Define
$\transflifdet\bigl(x,p,v\bigr) = \bigl(\transflif^v(x,p),v\bigr)$. Then, it is immediate to see that
$\invol \circ \transflifdet \circ \invol =
\transflifdet^{-1}$. We consider the deterministic proposal
kernel
\begin{equation}
  Q\bigl((x,p,v),\rmd(y,q,w)\bigr) = \updelta_{\transflif^v(x,p)}(\rmd (y,q)) \updelta_{v}(\rmd w) \eqsp.
\end{equation}
In the case, the
acceptance ratio \eqref{eq:generalized-MH-deterministic}
reads for $x,p\in\rset^d$,  $v \in \msv$ satisfying $ \pi_0(x)\varphi(p) > 0$% (y,q,w)= (\transflif^v(x,p),v)$
%{\small{
\begin{equation}
\label{eq:acceptance-probability-determinist-lif}
\bar{\alpha}(x,p,v) = \accfun\left( \mu\bigl(\transflif^v(x,p)\bigr) \Jac_{\transflif^v}(x,p)/{\mu(x,p)} \right) \eqsp,
\end{equation}
%}}
%
%\vspace{-0.6cm}
%\noindent
and is equal to $1$ if $\mu(x,p)=0$, where $\mu(x,p)= \pi_0(x)\varphi(p)$; see \Cref{subsec:proof_lifted_deterministic}.

\begin{example}[L2HMC]
  \label{ex:l2hmc}
  Assume that $\pi_0$ is positive and continuously differentiable.
  Using the framework depicted above, we show how the L2HMC algorithm~\cite{levy:hoffman:sohl-dickstein:2017} (Learning To Hamiltonian Monte Carlo) can be turned into a
  nonreversible MCMC method by considering the
  map
\begin{equation}
  \label{eq:def_lDhmc}
  \transflif(x,p) = G_{K}\circ\dots\circ G_{1}(x,p) \eqsp,
\end{equation}
where $G_i = H_{i} \circ F_{i} \circ H_{i-1/2}$ with, for $\delta>0$, 
\begin{enumerate}[wide, labelwidth=!, labelindent=0pt,label=$\bullet$,noitemsep,nolistsep]
\item for $j \in \{i,i-1/2\}$, $H_j(x,p)=(x, H_{j,x}(p))$ with
{\small  
  \begin{equation}
    H_{j,x}(p) =p \odot \exp\left(\delta R_{j}^H(x)\right) +\delta
  \bigl[\nabla \log \pi_0(x) \odot \exp\left(\delta R_{j}^H(x)\right)+
  M_{j}^H(x)\bigr]\eqsp.
  \end{equation}
  }
   Note that $H_j$ is a $\rmc^1$-diffeomorphism on
  $\rset^{2d}$ of Jacobian $\Jac_{H_j}(x,p)=\exp\left(\delta R_{j}^H(x)\right)$.
\item $F_i(x,p) = \bigl(F_{i,p}(x),p\bigr)$, where
  $F_{i,p}$ splits its input
into two parts $x_1,x_2$ and applies affine transformations
{\small\begin{equation}
\begin{aligned}
x_1'& = x_1 \odot \exp\bigl(\delta R^F_{i,1}(x_2,p)\bigr) + \delta M^F_{i,1}(x_2,p) \eqsp, \\
x_2'& = x_2 \odot \exp\bigl(\delta R^F_{i,2}(x_1',p)\bigr) + \delta M^F_{i,2}(x_1',p) \eqsp.
\end{aligned}
\end{equation}}
Clearly, $F_{i}$ is a $\rmc^1$-diffeomorphism on $\rset^{2d}$ with $\Jac_{F_i}(x,p) =\exp\bigl(\delta R_{i,1}^F(x_2,p)+\delta R_{i,2}^F(x_1',p)\bigr)$.
\end{enumerate}
Then, $\transflif$ defined by~\eqref{eq:def_lDhmc} is a
$\rmC^1$-diffeomorphism whose Jacobian can be recursively computed. Then, the kernel $P(x,p,w),\rmd (y,q,w))$ given by
\begin{equation}
 \bar{\alpha}(x,p,v)\updelta_{\transflif^v(x,p)}(\rmd (y,q)) \updelta_{v}(\rmd w) 
+ (1- \bar{\alpha}(x,p,v))\updelta_{(x,p,-v)}(\rmd (y,q,w))
\end{equation}
where $\bar{\alpha}$ is defined in \eqref{eq:acceptance-probability-determinist-lif} is $(\pi, S)$-reversible. This kernel should be combined with (possibly partial) refreshment steps as discussed in \Cref{sec:gener-hamilt-dynam}; see \Cref{subsec:l2hmc_implem} for details.
\end{example}

\section*{Acknowledgments}

AD and EM acknowledge support of the Lagrange Mathematical and Computing Research Center. Part of the paper was prepared within the framework of the HSE University Basic Research Program and funded by the Russian Academic Excellence Project -5-100.
It was also supported by the ANR ANR-19-CHIA-0002-01 “Chaire d’excellence en IA”, projet SCAI.

\newpage
\bibliographystyle{apalike2}
\bibliography{mcmc}

\begin{thebibliography}{}

\bibitem[Albergo et~al., 2019]{albergo:2019}
Albergo, M., Kanwar, G., \& Shanahan, P. (2019).
\newblock Flow-based generative models for {M}arkov chain {M}onte {C}arlo in
  lattice field theory.
\newblock {\em Physical Review D}, 100(3), 034515.

\bibitem[Andrieu \& Livingstone, 2019]{andrieu:livingstone:2019}
Andrieu, C. \& Livingstone, S. (2019).
\newblock Peskun-{T}ierney ordering for {M}arkov chain and process {M}onte
  {C}arlo: beyond the reversible scenario.
\newblock {\em Accepted for publication in Ann. Statist.}

\bibitem[Ardizzone et~al., 2019]{ardizzone:2019}
Ardizzone, L., Lüth, C., Kruse, J., Rother, C., \& Köthe, U. (2019).
\newblock Guided image generation with conditional invertible neural networks.

\bibitem[Baptista et~al., 2020]{baptista:2020}
Baptista, R., Zahm, O., \& Marzouk, Y. (2020).
\newblock An adaptive transport framework for joint and conditional density
  estimation.
\newblock {\em arXiv preprint arXiv:2009.10303}.

\bibitem[Bierkens \& Roberts, 2017]{bierkens:2017}
Bierkens, J. \& Roberts, G. (2017).
\newblock A piecewise deterministic scaling limit of lifted
  {M}etropolis–{H}astings in the {C}urie–{W}eiss model.
\newblock {\em Ann. Appl. Probab.}, 27(2), 846--882.

\bibitem[Chen et~al., 1999]{chen:lovasz:pak:1999}
Chen, F., Lov\'{a}sz, L., \& Pak, I. (1999).
\newblock Lifting {M}arkov chains to speed up mixing.
\newblock In {\em Annual {ACM} {S}ymposium on {T}heory of {C}omputing
  ({A}tlanta, {GA}, 1999)}  (pp.\ 275--281). ACM, New York.

\bibitem[Dalalyan, 2017]{dalalyan:2017}
Dalalyan, A.~S. (2017).
\newblock Theoretical guarantees for approximate sampling from smooth and
  log-concave densities.
\newblock {\em Journal of the Royal Statistical Society: Series B (Statistical
  Methodology)}, 79(3), 651--676.

\bibitem[Dalalyan \& Karagulyan, 2019]{dalalyan:2019}
Dalalyan, A.~S. \& Karagulyan, A. (2019).
\newblock User-friendly guarantees for the {L}angevin {M}onte {C}arlo with
  inaccurate gradient.
\newblock {\em Stochastic Processes and their Applications}, 129(12),
  5278--5311.

\bibitem[Diaconis et~al., 2000]{diaconis:holmes:neal:2000}
Diaconis, P., Holmes, S., \& Neal, R.~M. (2000).
\newblock Analysis of a nonreversible {M}arkov chain sampler.
\newblock {\em Ann. Appl. Probab.}, 10(3), 726--752.

\bibitem[Dinh et~al., 2017]{dinh:2016}
Dinh, L., Sohl-Dickstein, J., \& Bengio, S. (2017).
\newblock Density estimation using {R}eal {NVP}.

\bibitem[Douc et~al., 2018]{douc:moulines:priouret:2018}
Douc, R., Moulines, E., Priouret, P., \& Soulier, P. (2018).
\newblock {\em {M}arkov chains}.
\newblock Springer Series in Operations Research and Financial Engineering.
  Springer, Cham.

\bibitem[Duane et~al., 1987]{duane:1987}
Duane, S., Kennedy, A., Pendleton, B.~J., \& Roweth, D. (1987).
\newblock Hybrid {M}onte {C}arlo.
\newblock {\em Physics Letters B}, 195(2), 216 -- 222.

\bibitem[Durmus \& Moulines, 2017]{durmus:2017}
Durmus, A. \& Moulines, E. (2017).
\newblock Nonasymptotic convergence analysis for the unadjusted {L}angevin
  algorithm.
\newblock {\em The Annals of Applied Probability}, 27(3), 1551--1587.

\bibitem[Durmus et~al., 2017]{durmus:moulines:saksman:2017}
Durmus, A., Moulines, E., \& Saksman, E. (2017).
\newblock On the convergence of {H}amiltonian {M}onte {C}arlo.
\newblock {\em Accepted for publication in Ann. Statist.}

\bibitem[Fang et~al., 2014]{fang:2014}
Fang, Y., Sanz-Serna, J.-M., \& Skeel, R.~D. (2014).
\newblock Compressible generalized hybrid {M}onte {C}arlo.
\newblock {\em The Journal of chemical physics}, 140(17), 174108.

\bibitem[Gustafson, 1998]{gustafson:1998}
Gustafson, P. (1998).
\newblock A guided walk {M}etropolis algorithm.
\newblock {\em Statistics and computing}, 8(4), 357--364.

\bibitem[Horowitz, 1991]{horowitz:1991}
Horowitz, A.~M. (1991).
\newblock A generalized guided {M}onte {C}arlo algorithm.
\newblock {\em Physics Letters B}, 268(2), 247--252.

\bibitem[Hukushima \& Sakai, 2013]{hukushima2013irreversible}
Hukushima, K. \& Sakai, Y. (2013).
\newblock An irreversible {M}arkov-chain {M}onte {C}arlo method with skew
  detailed balance conditions.
\newblock In {\em Journal of Physics: Conference Series}, volume 473  (pp.\
  012012).: IOP Publishing.

\bibitem[Levy et~al., 2017]{levy:hoffman:sohl-dickstein:2017}
Levy, D., Hoffman, M.~D., \& Sohl-Dickstein, J. (2017).
\newblock Generalizing {H}amiltonian {M}onte {C}arlo with neural networks.
\newblock {\em arXiv preprint arXiv:1711.09268}.

\bibitem[Ma et~al., 2016]{ma2016unifying}
Ma, Y.-A., Chen, T., Wu, L., \& Fox, E.~B. (2016).
\newblock A unifying framework for devising efficient and irreversible {MCMC}
  samplers.
\newblock {\em arXiv preprint arXiv:1608.05973}.

\bibitem[Mengersen \& Tweedie, 1996]{mengersen:tweedie:1996}
Mengersen, K. \& Tweedie, R.~L. (1996).
\newblock Rates of convergence of the {H}astings and {M}etropolis algorithms.
\newblock {\em Ann. Statist.}, 24, 101--121.

\bibitem[Michel, 2016]{michel:2016}
Michel, M. (2016).
\newblock {\em Irreversible {M}arkov chains by the factorized {M}etropolis
  filter : algorithms and applications in particle systems and spin models}.
\newblock PhD thesis.
\newblock Thèse de doctorat dirigée par Krauth, Werner Physique Paris
  Sciences et Lettres (ComUE) 2016.

\bibitem[Neal, 2011]{neal:2011}
Neal, R.~M. (2011).
\newblock {MCMC} using {H}amiltonian dynamics.
\newblock {\em Handbook of {M}arkov Chain {M}onte {C}arlo}, (pp.\ 113--162).

\bibitem[Neklyudov et~al., 2020]{neklyudov:welling:vetrov:2020}
Neklyudov, K., Welling, M., Egorov, E., \& Vetrov, D. (2020).
\newblock Involutive {MCMC}: a unifying framework.
\newblock {\em arXiv preprint arXiv:2006.16653}.

\bibitem[Ottobre, 2016]{ottobre2016markov}
Ottobre, M. (2016).
\newblock {M}arkov chain {M}onte {C}arlo and irreversibility.
\newblock {\em Reports on Mathematical Physics}, 77(3), 267--292.

\bibitem[Papamakarios et~al., 2019]{papamakarios2019normalizing}
Papamakarios, G., Nalisnick, E., Rezende, D.~J., Mohamed, S., \&
  Lakshminarayanan, B. (2019).
\newblock Normalizing flows for probabilistic modeling and inference.
\newblock {\em arXiv preprint arXiv:1912.02762}.

\bibitem[Prangle, 2019]{prangle:2019}
Prangle, D. (2019).
\newblock Distilling importance sampling.
\newblock {\em arXiv preprint arXiv:1910.03632}.

\bibitem[Sherlock \& Thiery, 2019]{sherlock:2019}
Sherlock, C. \& Thiery, A.~H. (2019).
\newblock A discrete bouncy particle sampler.
\newblock {\em arXiv preprint 1707.05200}.

\bibitem[Sohl-Dickstein et~al., 2014]{sohl2014hamiltonian}
Sohl-Dickstein, J., Mudigonda, M., \& DeWeese, M.~R. (2014).
\newblock {H}amiltonian {M}onte {C}arlo without detailed balance.
\newblock {\em arXiv preprint arXiv:1409.5191}.

\bibitem[Song et~al., 2017]{song2017nice}
Song, J., Zhao, S., \& Ermon, S. (2017).
\newblock {A-NICE-MC}: Adversarial training for {MCMC}.
\newblock In {\em Advances in Neural Information Processing Systems}  (pp.\
  5140--5150).

\bibitem[Spanbauer et~al., 2020]{spanbauer:freer:2020}
Spanbauer, S., Freer, C., \& Mansinghka, V. (2020).
\newblock Deep involutive generative models for neural mcmc.
\newblock {\em arXiv preprint arXiv:2006.15167}.

\bibitem[Thin et~al., 2020]{thin:2020}
Thin, A., Kotelevskii, N., Denain, J.-S., Grinsztajn, L., Durmus, A., Panov,
  M., \& Moulines, E. (2020).
\newblock Metflow: A new efficient method for bridging the gap between {M}arkov
  chain {M}onte {C}arlo and variational inference.
\newblock {\em arXiv preprint arXiv:2002.12253}.

\bibitem[Tierney, 1994]{tierney:1994}
Tierney, L. (1994).
\newblock {M}arkov chains for exploring posterior distributions.
\newblock {\em Ann. Statist.}, 22(4), 1701--1728.

\bibitem[Tierney, 1998]{tierney:1998}
Tierney, L. (1998).
\newblock A note on {Metropolis-Hastings} kernels for general state spaces.
\newblock {\em Ann. Appl. Probab.}, 8(1), 1--9.

\bibitem[Turitsyn et~al., 2011]{turitsyn:2011}
Turitsyn, K.~S., Chertkov, M., \& Vucelja, M. (2011).
\newblock Irreversible {M}onte {C}arlo algorithms for efficient sampling.
\newblock {\em Physica D: Nonlinear Phenomena}, 240(4), 410 -- 414.

\end{thebibliography}

\appendix
\section{Notations, definitions and general Markov chain theory}
\label{sec:notat-defin}
In this section, we recall some basic facts and notations in a form that is useful for establishing properties of Markov chains. Let $(\Zset,\Zsigma)$ be a measurable space where $\Zsigma$ is a countably generated $\sigma$-algebra.

\begin{definition} [Kernel] \label{def:kernel}
A \emph{kernel} on $\Zset \times \Zsigma$ is a map $P\colon \Zset \times \Zsigma \to \rset_+$ such that
\begin{enumerate}[label=(\roman*)]
\item for any $ A \in \Zsigma$, $z \mapsto P(z,A)$ is measurable;
\item for any $z \in \Zset$, the function $A \mapsto P(z,A)$ is a finite measure on $\Zsigma$.
\end{enumerate}
\end{definition}
\begin{definition} [Markov and sub-Markovian kernel]\label{def:markovkernel}
A kernel $P$ is Markovian (or $P$ is a Markov kernel) if $P(z,\Zset)=1$ for all $z \in \Zset$.
A kernel $P$ is submarkovian (or $P$ is a sub-Markov kernel) if $ P(z,\Zset) \leq 1$ for all $z \in \Zset$.
\end{definition}

For $f\colon \Zset \rightarrow \mathbb{R}$ a measurable function, $\nu$ a probability distribution, and $P$  a kernel on $\Zset \times \Zsigma$, we let $\nu(f) \eqdef \int f(z) \nu(\rmd z)$ and denote for $(z,A) \in \Zset\times \Zsigma$,
\[
\nu P(A)= \int \nu(\rmd z) P(z,A) \eqsp, \quad P f(z) = \int P(z,\rmd z') f(z') \eqsp.
\]
Further, for $(z,A) \in \Zset\times \Zsigma$ define recursively for $n\geq 2$: $P^n(z,A)=\int P^{n-1}(z,\rmd z')P(z',A)$.

\begin{definition}[Total variation distance]
For $\mu,\nu$ two probability distributions on $(\Zset,\Zsigma)$ we define the total variation distance between $\mu$ and $\nu$ by $\|\mu-\nu\|_{\rm TV} := \sup_{|f|\leq 1} |\mu(f)-\nu(f)|$, where the supremum is taken over the measurable function $f\colon \Zset\rightarrow \mathbb{R}$.

\end{definition}

%If $\nu$ is a measure on $(\Zset,\Zsigma)$ and $h$ is a nonnegative $\Zset$-measurable function, then $\nu P$, $Ph$

\begin{definition} [Harmonic function]
Let $P$ be a kernel on $(\Zset,\Zsigma)$. Then a non-negative measurable function $h:  \Zset \to \rset$ is said to be \emph{harmonic} if $Ph=h$.
\end{definition}
\begin{definition}[Irreducibility] \label{def:irreducibility} Let
  $\nu$ be a non trivial $\sigma$-finite measure on $\big(\Zset,\Zsigma\big)$. A
   kernel $P$ is said to be $\nu$-irreducible if for all
  $(z,A)\in\Zset\times\Zsigma$ such that $\nu(A)>0$ there exists
  $n=n(z,A)\in\mathbb{N}$ such that $P^n(z,A)>0$.
\end{definition}

\begin{definition} [Periodicity and Aperiodicity]  \label{def:periodicity} $P$ is periodic if there exists $n \in \mathbb{N}$, $n\geq 2$, and $A_i \in \Zsigma$ for $i\in {1,\ldots,n}$, non-empty and disjoint, such that for $z \in A_i$, $P(z,A_{i+1})=1$ with the convention $A_{n+1}=A_1$. Aperiodicity is the negation of periodicity.
\end{definition}

General Markov chain theory provides us with powerful tools to establish validity and convergence of MCMC algorithms, leading to basic convergence theorems such as those found in~\cite[Theorem~1 and 3]{tierney:1994} and distilled below. We informally comment on the result below.
\begin{theorem} [\cite{tierney:1994}] \label{theo:tierney-supp}Suppose $P$ is such that $\pi P = P$ and is $\pi-$irreducible. Then $\pi$ is the unique invariant probability distribution of $P$ and for any $f\colon\Zset\rightarrow\mathbb{R}$ such that $\pi(|f|)<\infty$
 \begin{equation} \label{eq:averages-convergence}
 \lim_{n\rightarrow \infty} n^{-1} \sum_{i=1}^n f(Z_i) = \pi(f) \eqsp,
 \end{equation}
 almost surely for $\pi-$almost all $z\in\Zset$. If in addition $P$ is aperiodic then
  for $\pi-$almost all $z\in\Zset$
 \begin{equation} \label{eq:TVconvergence}
 \lim_{n\rightarrow \infty}\|P^n(z,\cdot) - \pi(\cdot)\|_{\rm TV} =0 \eqsp .
 \end{equation}
 \end{theorem}

 The result is fairly intuitive. Invariance of $\pi$ is a fixed point property ensuring that if $Z_i \sim \pi$ then $Z_{i+1}\sim \pi$. $\pi-$irreducibility simply says that the Markov chain should be able to reach any set of $\pi-$positive probability from any $z\in\Zset$ in a finite number of iterations. Periodicity would clearly prevent~\eqref{eq:TVconvergence} since the Markov chain would then periodically avoid visiting sets of positive $\pi-$probability. Averaging in~\eqref{eq:averages-convergence} removes the need for this property. We note that establishing these properties is often overlooked and a necessary prerequisite to any more refined analysis characterising their performance, such as quantitative finite time convergence bounds as found for example in~\cite{dalalyan:2017,dalalyan:2019,durmus:2017}.

\section{Standard reversible MH}
\label{sec:standard-reversible-MH}
We summarize in this Section the results presented in \cite[Section 2]{tierney:1998}.

\begin{definition} [Reversible kernel]
A sub-Markovian kernel $P$ on $(\Zset,\Zsigma)$, $P$ is $\pi$-reversible if and only if
\begin{equation}
  \label{eq:rev_def}
  \dnu\bigl(\rmd(z,z')\bigr) = \dnu^{\sym}\bigl(\rmd(z,z')\bigr) \eqsp,
\end{equation}
where $\dnu\bigl(\rmd(z,z')\bigr) = \pi(\rmd z) P(z,\rmd z')$ and
$\dnu^{\sym}\bigl(\rmd(z,z')\bigr) = \pushf{\sym}{\nu}\bigl(\rmd (z,z')\bigr) = \pi(\rmd z')
P(z',\rmd z)$ is the pushfoward measure of $\nu$ by $\sym\colon (z,z') \mapsto (z',z)$.
\end{definition}

% Our analysis is an
% extension of the study provided by \cite[Section
% 2]{tierney:1998}.
From a proposal Markov kernel $Q$, the MH method consists of
considering a sub-Markovian kernel $Q_{\alpha}(z,\rmd z')= \alpha(z,z') Q(z,\rmd z')$.
If $\pi$
and $Q$ admit a common dominating $\sigma$-finite measure $\mu$ on $\msz$, such that $\pi(\rmd z) = \pi(z) \mu(\rmd z)$ (we use the same notation for the probability and the density) and $Q(z,\rmd z')= q(z,z') \mu(\rmd z')$, $Q_{\alpha}$ is $\pi$-reversible if
\[
\alpha(z,z') =
\begin{cases}
\accfun\left(\frac{\pi(z')q(z',z)}{\pi(z)q(z,z')}\right) & \pi(z) q(z,z') > 0 \eqsp, \\
1 & \text{otherwise} \eqsp,
\end{cases}
\]
where for any $t \in \rset_+^*$,
\begin{equation}
  \label{eq:prop_accfun_supp}
    t\accfun(1/t) = \accfun(t) \eqsp.
  \end{equation}
  %leads to reversibility of $Q_{\alpha}$ \wrt~$\pi$.
 We may take for example  $\accfun(t)= \min(1,t)$ or $\accfun(t)= t/(1+t)$
which correspond to the classical Metropolis-Hastings and Barker ratio, respectively.
To obtain a $\pi$-reversible Markov kernel $P$, it suffices to add a Dirac mass, \ie~
\begin{equation}
  \label{eq:2}
  P(z,\rmd z') = Q_{\alpha}(z,\rmd z') + \bigl(1 - Q_{\alpha}(z,\msz)\bigr) \updelta_{z}(\rmd z') \eqsp.
\end{equation}
This construction can be generalized to the case where $\pi$ or $Q$ do not admit a density. In particular, let $\transfdet$ be an invertible mapping on $\Zset$ satisfying $\transfdet^{-1}= \transfdet$ (\ie\ $\transfdet$ is an involution) and consider  $Q(z,\rmd z')= \updelta_{\transfdet(z)}(\rmd z')$ (when the current state is $z$, then the proposal is $\transfdet(z)$). Define the measure $\nu= \pi + \pushf{\transfdet}{\pi}$ and denote by $h(z)= \rmd \pi/ \rmd \nu(z)$ ($h$ is the density of $\pi$ \wrt\ $\nu$).  Then, $h\bigl(\transfdet(z)\bigr)$ is a density of $\pushf{\transfdet}{\pi}$ \wrt\ $\nu$. Denote $A= \set{z \in \Zset}{h(z) \times h \circ \transfdet(z) > 0}$. Detailed balance holds if and only if for $\pi$-almost all $z \in A$ (see \cite{tierney:1998}):
\[
\alpha\bigl(z,\transfdet(z)\bigr) h(z)/ h \circ \transfdet(z)= \alpha\bigl(\transfdet(z), z\bigr) \eqsp.
\]
If $\Zset= \rset^d$ and $\nu$ is the Lebesgue measure, we obtain $\alpha\bigl(z,\transfdet(z)\bigr)= \bar{\alpha}(z)$, where
\[
\bar{\alpha}(z) = \accfun\left( \frac{\pi \circ \transfdet(z)}{\pi(z)} \Jac_{\transfdet}(z) \right) \eqsp.
\]

\section{Proofs of \Cref{sec:mu-s-reversibility}}
\subsection{Proof of \eqref{eq:condition-reversibility}}
\label{sec:proof:condition-reversibility}
Let $f\colon \Zset^2 \to \rset_+$ be a measurable function. The condition $\dmu_P = \dmu^\invol_P$ implies
\[
I= \iint \dmu_P\bigl(\rmd (z,z')\bigr) f(z,z')= \iint \dmu_P\bigl(\rmd (z,z')\bigr) f\bigl(\invol(z'),\invol(z)\bigr) = \iint \pi(\rmd z) P(z,\rmd z') f\bigl(\invol(z'),\invol(z)\bigr) \eqsp.
\]
Using the change of variable $\tilde{z}'= \invol(z)$ and since $\invol$ is an involution, we get
\[
I= \iint \pushf{s}{\pi}(\rmd \tilde{z}') P\bigl(\invol(\tilde{z}'), \rmd z'\bigr) f\bigl(\invol(z'),\tilde{z}'\bigr) \eqsp.
\]
Applying now the change of variable $\tilde{z}= \invol(z')$, we finally obtain
\[
I= \iint \pushf{s}{\pi}(\rmd \tilde{z}') \pushf{s}{P}\bigl(\invol(\tilde{z}'), \rmd \tilde{z}\bigr) f(\tilde{z},\tilde{z}') \eqsp.
\]
Note that, for any $z \in \Zset$ and $A \in \Zsigma$,
\[
\pushf{\invol}{P}(z,A) = \int P(z, \rmd z') \indi{A}\bigl(\invol(z')\bigr)= P \involk(z,A) \eqsp,
\]
showing that
\[
I = \iint \pushf{s}{\pi}(\rmd \tilde{z}') P\involk\bigl(\invol(\tilde{z}'),\rmd \tilde{z}\bigr) f(\tilde{z},\tilde{z}')
= \iint \pushf{s}{\pi}(\rmd \tilde{z}') \involk P\involk(\tilde{z}',\rmd \tilde{z}) f(\tilde{z},\tilde{z}'),
\]
where we have used $\int \involk P g(\tilde{z}',\rmd \tilde{z})= P g\bigl(\invol(\tilde{z}')\bigr)$.

\subsection{Proof of \Cref{prop:extension-tierney}}
We set $\dlambda= \dnu + \dnu^\invol$. Note that $\dnu$ and $\dnu^\invol$ are absolutely continuous \wrt\ to $\dlambda$.
Denote by $\dlambda^\invol = \pushf{(\sym_{\invol})}{\dlambda}$ the
pushforward of $\dlambda$ by the transform
$\sym_{\invol}(z,z')= \bigl(\invol(z'),\invol(z)\bigr)$: for any
 $C \in \Zsigma^{\otimes 2}$
\begin{equation}
\label{eq:s-symmetrization}
  \dlambda^\invol(C)= \int \indi{C}\bigl(\invol(z'),\invol(z)\bigr) \dlambda\bigl(\rmd(z,z')\bigr) \eqsp.
\end{equation}
Since $(\dnu^\invol)^{\invol}= \dnu$,  $\dlambda= \dlambda^\invol$. This implies, for any measurable function $f\colon \Zset^2 \to \rset_+$,
\begin{equation}
\label{eq:symmetry-lambda}
  \iint f(z,z') \dlambda\bigl(\rmd(z,z')\bigr)= \iint f\bigl(\invol(z'),\invol(z)\bigr) \dlambda \bigl(\rmd(z,z')\bigr) \eqsp.
\end{equation}
We choose $h$ to be a version of the Radon-Nikodym derivative
$\rmd \dnu / \rmd \dlambda$ (the function is defined up to
$\dlambda$-negligible sets).  Then by definition of $\dnu^\invol$,
\begin{align}
  \iint f(z,z') \dnu^\invol\bigl(\rmd(z,z')\bigr)
  &= \iint f\bigl(\invol(z'),\invol(z)\bigr) \dnu\bigl(\rmd(z,z')\bigr)= \iint f\bigl(\invol(z'), \invol(z)\bigr) h(z,z') \dlambda \bigl(\rmd(z,z')\bigr) \\
  &= \iint f\bigl(\invol(z'), \invol(z)\bigr) h(z,z') \dnu \bigl(\rmd(z,z')\bigr) + \iint f\bigl(\invol(z'), \invol(z)\bigr) h(z,z') \dnu^\invol \bigl(\rmd(z,z')\bigr) \\
  &= \iint f(z,z') h\bigl(\invol(z'),\invol(z)\bigr) \dnu^{\invol}\bigl(\rmd(z,z')\bigr)+ \iint f(z,z') h\bigl(\invol(z'),\invol(z)\bigr) \dnu\bigl(\rmd(z,z')\bigr) \\
  &= \iint f(z,z') h\bigl(\invol(z'),\invol(z)\bigr) \dlambda\bigl(\rmd(z,z')\bigr) \eqsp,
\end{align}
showing that
\begin{equation}
\label{eq:tierney-1}
  h\bigl(\invol(z'),\invol(z)\bigr) = \frac{\rmd \dnu^\invol}{\rmd \dlambda}(z,z') \eqsp.
\end{equation}
We then define
\begin{equation}
\label{eq:definition-A}
  A_{\dnu} = \defEns{ (z,z') \in \Zset^2\colon h(z,z') \times h\bigl(\invol(z'),\invol(z)\bigr) > 0 } \eqsp.
\end{equation}
In other words, if $(z,z') \not \in A_{\dnu}$, then either $h(z,z')= 0$ or $h\bigl(\invol(z'),\invol(z)\bigr)=0$. Therefore, $\dnu_{A,\comp}$ and $\dnu_{A,\comp}^\invol$ are singular since $B_1 = \defEnsLigne{ (z,z') \in \Zset^2\colon h(z,z') > 0 }$, $B_2 = \defEnsLigne{ (z,z') \in \Zset^2\colon h\bigl(\invol(z'),\invol(z)\bigr) > 0 }$ are disjoint subsets of $A_{\dnu}^{\comp}$ and $\dnu(B_2) = 0$, $\dnu(B_1) = 0$. In addition, since for any set $B \in \Zset^2$,
\[
  \indi{A_{\dnu} \cap B} h =0 \; \dlambda-\mae \text{ if and only if } \indi{A_{\dnu} \cap B} h^\invol =0 \; \dlambda-\mae
\]
the restrictions $\dnu_A$ and $\dnu_A^\invol$ are equivalent. In addition,
\begin{equation}
\label{eq::ratio-density}
  \frac{\rmd \dnu_A}{\rmd \dnu^\invol}(z,z') = \frac{h(z,z')}{h\bigl(\invol(z'),\invol(z)\bigr)} = r(z,z') \eqsp, \quad (z,z') \in A_{\dnu} \eqsp,
\end{equation}
satisfying $r(z,z')= 1 / r\bigl(\invol(z'),\invol(z)\bigr)$.
%It is easily shown that $A$ is unique up to $\dlambda$-negligible sets.

\subsection{Proof of \Cref{theo:extension-tierney}}

Define the $\sigma$-finite measure $\drho\bigl(\rmd(z,z')\bigr)= \alpha(z,z') \dnu\bigl(\rmd(z,z')\bigr)$ and denote by $\drho^\invol = \pushf{(\sym_{\invol})}{\drho}$ the
pushforward of $\drho$ by the transform $\sym_{\invol}(z,z')= \bigl(\invol(z'),\invol(z)\bigr)$: for any $C \in \Zsigma^{\otimes 2}$
\begin{equation}
\label{eq:s-symmetrization}
  \drho^\invol(C) = \int \indi{C}\bigl(\invol(z'),\invol(z)\bigr) \drho\bigl(\rmd(z,z')\bigr) \eqsp.
\end{equation}
Note by definition of $\dnu^\invol$ that
\begin{equation}
\label{eq:proof_theo_drho}
  \drho^\invol\bigl(\rmd(z,z')\bigr) = \alpha\bigl(\invol(z'),\invol(z)\bigr) \dnu^\invol\bigl(\rmd(z,z')\bigr)
\end{equation}
We show below that under the stated assumptions $\drho = \drho^\invol$.

Define the function $\tilde{\alpha}(z,z') = \alpha\bigl(\invol(z'),\invol(z)\bigr)$.
Since the set $A_{\dnu}$ is $\invol$-symmetric, the set $A_{\dnu}^{\comp}$ is also $\invol$-symmetric and $\dnu( \set{(z,z') \in A_{\dnu}^{\comp}}{\tilde{\alpha}(z,z') > 0})=0$ using \ref{item:extension-tierney-i}.
Hence $\drho(A_{\dnu}^{\comp})= \drho^\invol(A_{\dnu}^{\comp})=0$.

We have by \Cref{prop:extension-tierney} and~\ref{item:extension-tierney-ii},
\begin{align}
\label{eq:condition-equality}
  \indi{A_{\dnu}}(z,z') \drho\bigl(\rmd(z,z')\bigr)
  &= \indi{A_{\dnu}}(z,z') \alpha(z,z') \dnu \bigl(\rmd(z,z')\bigr)= \alpha(z,z') r(z,z')\dnu^{\invol} \bigl(\rmd(z,z')\bigr)\\
  \nonumber
  &= \indi{A_{\dnu}}(z,z') \alpha\bigl(\invol(z'),\invol(z)\bigr) \nu^{\invol}\bigl(\rmd(z,z')\bigr) = \indi{A_{\dnu}}(z,z') \drho^\invol\bigl(\rmd(z,z')\bigr) \eqsp.
\end{align}

Conversely, assume that $\drho= \drho^\invol$. Since by
\Cref{prop:extension-tierney}, $\dnu_{A,\comp}$ and
$\dnu_{A,\comp}^{\invol}$ are mutually singular, there exist
$B_1,B_2 \subset A_{\dnu}^{\comp}$ (see also the proof
\Cref{prop:extension-tierney}) such that $\dnu_{A,\comp}(B_2)= 0$ and
$\dnu_{A,\comp}^{\invol}(B_1)= 0$. Therefore, we obtain using that $\drho= \drho^\invol$  and \eqref{eq:proof_theo_drho} that
\begin{equation}
  \drho(A_{\dnu}^{\comp} \cap B_1) = \drho^{\invol}(A_{\dnu}^{\comp} \cap B_1)= 0 \eqsp.
\end{equation}
This result and $\drho(A_{\dnu}^{\comp} \cap B_2)$ imply $\drho(A_{\dnu}^{\comp})= 0$ and therefore
$\dnu(\set{(x,z') \in A_{\dnu}^{\comp}}{\alpha(z,z') > 0})=0$ showing
\ref{item:extension-tierney-i}. Finally, under the condition
$\drho = \drho^{\invol}$ and \Cref{prop:extension-tierney},
\eqref{eq:condition-equality} holds and~\ref{item:extension-tierney-ii} follows.

\subsection{Checking the GMH rule~\eqref{eq:generalized-MH}}
\label{sec:proof:generalized-MH}
We first check~\ref{item:extension-tierney-i}.
By \Cref{prop:extension-tierney} and~\eqref{eq:generalized-MH},  $A^c_\nu = B_1 \cup B_2$ where $B_1 = \set{(z,z') \in \Zset^2}{h(z,z')=0}$ and $B_2=\set{(z,z') \in \Zset^2}{h\bigl(\invol(z'),\invol(z)\bigr)=0}$, and for any $(z,z') \in B_2 \setminus B_1$, $\alpha(z,z')=0$. Therefore, to show~\ref{item:extension-tierney-i}, it suffices to establish that $\dnu(\{\alpha = 0\} \cap B_1) = 0$ which follows from
\[
\dnu(B_1)= \int \indi{B_1}(z,z') \dnu\bigl(\rmd (z,z')\bigr) = \int \indi{B_1}(z,z') h(z,z') \dlambda\bigl(\rmd (z,z')\bigr)= 0 \eqsp.
\]
%Therefore, $\dnu(\set{(z,z') \in \Zset^2}{\alpha(z,z')> 0}) = \dnu(B_1)= 0$ and \ref{item:extension-tierney-i} is satisfied.

We now check~\ref{item:extension-tierney-ii}. Note that by \Cref{prop:extension-tierney} and using that $\sym_{\invol}$ is an involution, for $(z,z') \in A_{\dnu}$,
\begin{align}
\alpha(z,z') r(z,z')&= \accfun\bigl(1/r(z,z')\big) r(z,z') \\
&=\accfun\bigl(r(z,z')\bigr) = \alpha\bigl(\invol(z'),\invol(z)\bigr) \eqsp.
\end{align}

\subsection{Expressions for $a$ and $b$}
\label{sec:proof:expressions-a-b}
We check the conditions on the nonnegative weights $a$ and $b$ so that the sub-Markovian kernel
\begin{equation}
  \label{eq:def_R_exp_a_b}
  R(z,\rmd z')= a(z) \updelta_{z}(\rmd z') + b(z) \updelta_{\invol(z)}(\rmd z')
\end{equation}
is $(\pi,\involk)$-reversible.
For $f$ a nonnegative measurable function, we get
\[
\involk R \involk f(z')= a\bigl(\invol(z')\bigr) f(z') + b\bigl(\invol(z')\bigr) f\bigl(\invol(z')\bigr) \eqsp.
\]
Hence, we obtain, for any nonnegative measurable function $g$,
\begin{align}
&\iint \pi(\rmd z') \involk R \involk (z',\rmd z) f(z) g(z')
= \int \pi( \rmd z') a\bigl(\invol(z')\bigr) f(z') g(z') + \int \pi(\rmd z') b\bigl(\invol(z')\bigr) f\bigl(\invol(z')\bigr) g(z') \\
&\qquad \qquad = \int \pi(\rmd z) a\bigl(\invol(z')\bigr) \updelta_{z}(\rmd z') f(z) g(z') +
  \int \pi(\rmd z') b(z') f(z') g\bigl(\invol(z')\bigr)\\
  &\qquad \qquad = \int \pi(\rmd z) a\bigl(\invol(z')\bigr) \updelta_{z}(\rmd z') f(z) g(z') +
  \int \pi(\rmd z) b(z) f(z)\updelta_{\invol(z)}(\rmd z') g(z')\eqsp,
\end{align}
where we have used $\pushf{\invol}{\pi}= \pi$. The result implies that
\[
\pi(\rmd z') SRS(z',\rmd z)= \pi(\rmd z)a\bigl(\invol(z)\bigr) \updelta_z(\rmd z') + \pi(\rmd z) b(z) \updelta_{\invol(z)}(\rmd z').
\]
Therefore, \eqref{eq:condition-reversibility} is satisfied (\eg~$\pi(\rmd z') SRS(z',\rmd z) = \pi(\rmd z) R(z,\rmd z')$) and $R$ is $(\pi,\involk)$-reversible if $a(z)= a\bigl(\invol(z)\bigr)$.

In addition, the total mass of $Q_\alpha(z,\rmd z')$ is $Q_\alpha(z,\Zset)$. The missing mass is  therefore $1 - Q_\alpha(z,\Zset)$. Since the total mass of $R$ is $a(z)+ b(z)$ we must have $a(z)+b(z)= 1 - Q_\alpha(z,\Zset)$.

We may for example set $a(z)=0$ and $b(z)= 1- Q_\alpha(z,\msz)$, which coincides with the classical MH rule when $\invol= \Id$. We may also take $a(z)= 1 - Q_\alpha(z,\Zset) - b(z)$ where $b$ satisfies $0 \leq b(z) \leq 1 - Q_\alpha(z,\Zset)$ and $b(z)- b\bigl(\invol(z)\bigr) = Q_\alpha(\invol(z),\Zset) - Q_\alpha(z,\Zset)$. As suggested in~\cite{turitsyn:2011}, we may set $b(z)= \max\bigl(0, Q_\alpha\bigl(\invol(z),\Zset\bigr) - Q_\alpha(z,\Zset)\bigr)$ which is shown to be optimal \wrt\ to the Peskun ordering in~\cite{andrieu:livingstone:2019}. Note however that this choice for $b$  is not always easily computable.

\subsection{Applications of~\eqref{eq:generalized-MH}: case with densities}
\label{sec:proof:proposal-with-densities}
%\subsubsection{Case with densities}
Note that  $\dnu\bigl(\rmd (z,z')\bigr)= \pi(z) q(z,z') \mu^{\otimes 2}\bigl(\rmd(z,z')\bigr)$,
$\dnu^\invol\bigl(\rmd (z,z')\bigr) = \pi\bigl(\invol(z')\bigr) q\bigl(\invol(z'),\invol(z)\bigr) \mu^{\otimes 2}\bigl(\rmd(z,z')\bigr)$,  since $\pushf{\invol}{\mu} = \mu$. In addition, $h = \tilde{h}/\{\tilde{h} + \tilde{h} \circ \sym_{\invol}\}$, $\tilde{h}(z,z') = \pi(z) q(z,z')$ and therefore $A_{\dnu}$ in~\eqref{eq:def_A_dnu} is given by
{\[
A_{\dnu} = \defEns{\pi(z) q(z,z') \times \pi\bigl(\invol(z')\bigr) q\bigl(\invol(z'),\invol(z)\bigr) > 0} \eqsp,
\]}
and for $(z,z') \in A_{\dnu}$,
\[
r(z,z')= \frac{\pi(z) q(z,z')}{\pi\bigl(\invol(z')\bigr) q\bigl(\invol(z'),\invol(z)\bigr)} \eqsp.
\]
Therefore, we obtain using~\eqref{eq:generalized-MH} that
\begin{equation}
\label{eq:3}
  \alpha(z,z') =
  \begin{cases}
    \accfun\parentheseDeux{\frac{\pi\bigl(\invol(z')\bigr) q\bigl(\invol(z'),\invol(z)\bigr)}{\pi(z) q(z,z')}} &  \pi(z) q(z,z') \ne 0, \\
    1  & \pi(z) q(z,z') = 0 \eqsp.
  \end{cases}
\end{equation}
In addition, note that using $\pushf{\invol}{\pi} = \pi$, $\pushf{\invol}{\mu} = \mu$ and $s$ is an involution, we obtain that
\begin{equation}
  \label{eq:egalite_pi_s_pi}
  \pi = \pi \circ \invol \eqsp.
\end{equation}
Therefore, we obtain
\begin{equation}
\label{eq:alpha_densite_alternative}
  \alpha(z,z') =
  \begin{cases}
    \accfun\parentheseDeux{\frac{\pi(z') q\bigl(\invol(z'),\invol(z)\bigr)}{\pi(z) q(z,z')}} & \pi(z) q(z,z') \ne 0,\\
    1   & \pi(z) q(z,z') = 0 \eqsp.
  \end{cases}
\end{equation}

\subsection{Proof of \Cref{theo:irred_and_co}}
\label{sec:proof:theo:irred_and_co}
We preface the proof by the following result. Define $\msz^+ = \{z \in \msz\colon \pi(z) >0\}$ and set
\begin{equation}
\label{eq:kernel-with-mass-a=0}
P(z,\rmd z')= Q_\alpha(z,\rmd z') + \{1 - Q_\alpha(z,\Zset)\} \updelta_{\invol(z)}(\rmd z')  \eqsp.
\end{equation}
where $Q_\alpha(z,z') = \alpha(z,z') Q(z,\rmd z')$ and $\alpha$ is given by~\eqref{eq:generalized-MH}. Note that $P$ corresponds to~\eqref{eq:kernel-with-mass} with $a \equiv 0$ and $b(z) = 1-Q_\alpha(z,\Zset)$.
\begin{proposition}
  Consider $P$ defined by \eqref{eq:kernel-with-mass-a=0}. Assume that $P$ is
  $\pi$-irreducible and $Q(z,\msz^+) = 1$ for any $z \not \in \msz^+$. Further, suppose that $\pi$ is not a Dirac mass. Then, $P$ is Harris recurrent.
\end{proposition}

\begin{proof}
  Since $\pi$ is invariant for $P$ by \Cref{theo:extension-tierney},
  $P$ is recurrent by~\cite[Theorem
  10.1.6]{douc:moulines:priouret:2018}. Therefore, \cite[Corollary 9.2.16, Proposition 5.2.12]{douc:moulines:priouret:2018} show that for any bounded
  harmonic function $h\colon \msz \to \rset$, \ie~satisfying $Ph = h$,
  $h = \pi(h)$, $\pi$~\mae. Then, if $A_h = \{h \neq \pi(h)\}$,
  $\pi(A_h)=0$. By~\cite[Theorem 10.2.11]{douc:moulines:priouret:2018}, $P$ is Harris recurrent if
  \begin{equation}
  \label{eq:to_show_harris}
  \text{ $h(z) = \pi(h)$ for any
  $z \in\msz$} \eqsp.
  \end{equation}

  First, consider $z \in\msz^+$. Define
  $B=\{z'\colon \pi(z') q\bigl(\invol(z'),\invol(z)\bigr) > 0\}$ and
  $C=\{z'\colon q(z,z') = 0\}$. Let $A$ be a $\pi$-negligible set, $\pi(A)=0$.
  Using $\accfun(t) \leq t$ by~\eqref{eq:prop_accfun}  for any $t \in \rset_+^*$,
  $\pi(z) \neq 0$ and~\eqref{eq:alpha_densite_alternative}, we get
  \begin{align}
  \label{eq:key-identity}
    \int \1_{A}(z') \alpha(z,z') q(z,z') \mu (\rmd z') &= \int \1_{A\cap B \cap C}(z') \alpha(z,z') q(z,z')   \mu (\rmd z')\\
    \nonumber
    & + \int \1_{A \cap B \cap C^{\comp}}(z')\accfun\parentheseDeux{\frac{\pi(z') q\bigl(\invol(z'),\invol(z)\bigr)}{\pi(z) q(z,z')}} q(z,z')   \mu (\rmd z') \\
    \nonumber
    & \leq \int \indi{A \cap B \cap C^{\comp}}(z')\frac{\pi(z') q\bigl(\invol(z'),\invol(z)\bigr)}{\pi(z)} \mu(\rmd z')\\
    \nonumber
    & \leq \int \indi{A \cap B \cap C^{\comp}}(z')\frac{q\bigl(\invol(z'),\invol(z)\bigr)}{\pi(z)}  \pi(\rmd z') =0 \eqsp,
  \end{align}
  where the last identity follows from $\pi(A)=0$. Applying this identity with $A_h$ yields to
  \begin{equation}
  \label{eq:nice-identity}
  \int \alpha(z,z') q(z,z') h(z') \rmd \mu(z') = \int \indi{A_h^{\comp}}(z') \alpha(z,z') q(z,z') h(z') \rmd \mu(z') = \pi(h) Q_{\alpha}(z,\msz) \eqsp.
  \end{equation}
  Therefore, the condition $Ph(z) =h(z)$ for any $z \in \Zset$ and \eqref{eq:kernel-with-mass-a=0}  imply that
  \begin{equation}
    \label{eq:eq_h}
    h(z) = \pi(h) Q_{\alpha}(z,\msz) + h \circ \invol(z) \{1-Q_{\alpha}(z,\msz)\} \eqsp.
  \end{equation}
  Applying $P$ to the previous equation, we obtain, denoting $\bar{\alpha}(z)= Q_\alpha(z,\Zset)$
  \begin{multline}
  \label{eq:rel0}
  h(z) = Ph(z) = \pi(h) \{ Q_\alpha \bar{\alpha}(z) + \{ 1 - \bar{\alpha}(z)\} \bar{\alpha} \circ \invol(z) \} \\
  + \int Q_\alpha(z,\rmd z') h \circ \invol(z') \{ 1 - \bar{\alpha}(z') \} + h(z) \{ 1 - \bar{\alpha}(z) \} \bigl\{1 - \bar{\alpha}\bigl(\invol(z)\bigr)\bigr\} \eqsp.
  \end{multline}
  Denote $A_{h \circ \invol}= \set{z \in \Zset}{h \circ \invol(z) \ne \pi(h)}$. Note that, $\pi(A_{h \circ \invol})= \pushf{\invol}{\pi}(A_h)= \pi(A_h)=0$. Using \eqref{eq:key-identity}, we get for $z \in \Zset^+$, $Q_\alpha(z,A_{h \circ \invol})=0$, which implies
  \begin{align}
  \int Q_\alpha(z,\rmd z') h \circ \invol(z') \{ 1 - \bar{\alpha}(z') \}
  &= \int \indi{A_{h \circ \invol}^c}(z') Q_\alpha(z,\rmd z') h \circ \invol(z') \{ 1 - \bar{\alpha}(z') \} \\
  &= \pi(h) \int  Q_\alpha(z,\rmd z')  \{ 1 - \bar{\alpha}(z') \} \eqsp.
  \end{align}
  Plugging this relation into~\eqref{eq:rel0} we obtain
  \begin{equation}
  h(z) = \pi(h) [ Q_\alpha \bar{\alpha}(z) + \{1 - \bar{\alpha}(z) \}\bar{\alpha} \circ \invol(z) ]
  + \pi(h) [\bar{\alpha}(z) - Q_\alpha \bar{\alpha}(z)] + h(z) \{1 - \bar{\alpha}(z)\} \{1 - \bar{\alpha} \circ \invol(z) \} \eqsp.
  \end{equation}
  Using straightforward algebra, the previous identity implies
  \[
  \{ \pi(h) - h(z) \} \{\bar{\alpha}(z) + \bar{\alpha} \circ \invol(z) - \bar{\alpha}(z) \times \bar{\alpha} \circ \invol(z) \} = 0 .
  \]
  Since $P$ is $\pi$-irreducible and $\pi$ is not a Dirac mass, $\bar{\alpha}(z) \neq 0$, we get that
  for all $z \in \msz^+$,
  \begin{equation}
    \label{eq:first_h_z_plus_eg}
    h(z) = \pi(h) \eqsp.
  \end{equation}

  Consider now the case $z \not \in \msz_+$. Using that $Q(z,\msz_+)$ by
  assumption and $\alpha(z,z')=1$ by
  \eqref{eq:alpha_densite_alternative} for any $z' \in \msz$, we get
  \begin{equation}
    %\label{eq:4}
    h(z) = Ph(z) = \int_{\msz_+} q(z,z') h(z') \mu(z') = \int_{\msz_+} \{q(z,z') h(z')/\pi(z')\} \pi(z') \mu(z') = \pi(h) \eqsp.
  \end{equation}
  Combining this result with \eqref{eq:first_h_z_plus_eg} completes the proof of~\eqref{eq:to_show_harris}.
\end{proof}

\begin{proposition} \label{prop:one-step-irreducibility}
Assume the conditions of \Cref{theo:irred_and_co}. Then for any $A \in \mcz$ such that $\pi(A) >0$, we have
  \begin{equation}
    \label{eq:irred_proof}
    P(z,A) >0 \quad \text{for any $z \in \Zset $} \eqsp.
  \end{equation}
\end{proposition}
\begin{proof}
  Consider first the case $z \in\msz_+ = \set{z \in \Zset}{\pi(z) > 0}$. Then, by~\eqref{eq:alpha_densite_alternative} and the condition if
  $\pi(z') >0$, then $q(\tilde{z},z') \times q\bigl(\invol(\tilde{z}),\invol(z')\bigr) >0$ for any
  $\tilde{z} \in \msz$, we have
  \begin{equation}
  %  \label{eq:5}
    P(z,A) \geq \int \1_{A\cap \msz_+}(z') \accfun\parentheseDeux{\frac{\pi(z') q\bigl(\invol(z'),\invol(z)\bigr)}{\pi(z) q(z,z')}} q(z,z')  \mu (\rmd z') >0 \eqsp,
  \end{equation}
  since $\pi(A) > 0$ implies that $\mu(A \cap \msz_+) > 0$.
  Second consider the case $z \not \in \msz_+$. Then, $\alpha(z,z')=1$ for any $z'\in\msz$ and we get
  \begin{equation}
    %\label{eq:2}
    P(z,A) \geq  \int \1_{A\cap \msz_+}(z') q(z,z')  \mu (\rmd z') > 0 \eqsp,
  \end{equation}
  which concludes the proof of~\eqref{eq:irred_proof}.
\end{proof}

\begin{proof}[Proof of \Cref{theo:irred_and_co}]
  $\pi$-irreducibility of $P$ follows from Proposition~\ref{prop:one-step-irreducibility}.  We show that $P$ is $\pi$-irreducible and aperiodic. Indeed, this result and~\cite[Theorem 7.2.1, Theorem 11.3.1]{douc:moulines:priouret:2018} imply~\eqref{eq:TVconvergence} for all $z \in \Zset$. Finally, \cite[Corollary 9.2.16, Proposition 5.2.14]{douc:moulines:priouret:2018} establish~\eqref{eq:averages-convergence} for all $z \in \Zset$.

  The fact that $P$ is aperiodic is a direct consequence of~\eqref{eq:irred_proof} and~\cite[Theorem 9.3.10]{douc:moulines:priouret:2018}.
\end{proof}

%Similarly, we obtain
%\begin{align}
%\int \dnu^{\invol}\bigl(\rmd (z,z')\bigr) f(z,z')
%&= \int \pi(\rmd z') f(T^{-1}(z'),z')\\
%&= \int k(z') f(T^{-1}(z'),z') \lambda(\rmd z') = \int k(T(z)) f(z,T(z)) \pushf{T^{-1}}{\lambda}(\rmd z) \\
%&= \int k(T(z)) f(z,T(z)) \frac{ \rmd \pushf{T^{-1}}{\lambda}}{\rmd \lambda} (z) \lambda(\rmd z) \eqsp.
%\end{align}
%On the other hand,
%\begin{align}
%\int f(z,z') \dnu^{\invol}\bigl(\rmd (z,z')\bigr)
%&= \int f(z,z') h(\invol(z),\invol(z')) \dlambda\bigl(\rmd (z,z')\bigr) \\
%&= \int f(z,T(z)) h(\invol \circ T(z), \invol(z)) \lambda(\rmd z) \eqsp.
%\end{align}
%Therefore, $\lambda$-\mae, we obtain
%\[
%k( T(z)) \frac{\rmd \pushf{T^{-1}}{\lambda}}{\rmd \lambda}(z)= h(s \circ T(z), \invol(z)) \eqsp.
% \]

\subsection{Proofs of~\eqref{eq:definition-k-lambda} and~\eqref{eq:generalized-MH-deterministic}}
\label{sec:proof:deterministic-case}

Consider $\dnu\bigl(\rmd(z,z')\bigr)= \pi(\rmd z) \updelta_{\transfdet(z)}(\rmd z')$, where $\transfdet$  is an invertible mapping on $\Zset$ satisfying $\transfdet^{-1}= \invol \circ \transfdet \circ \invol$. For any measurable function $f\colon \Zset^2 \to \rset_+$, we get
\begin{align}
\iint \dnu^{\invol}\bigl(\rmd (z,z')\bigr) f(z,z')
&= \iint \dnu\bigl(\rmd (z,z')\bigr) f\bigl(\invol(z'),\invol(z)\bigr) = \int \pi(\rmd z) f\bigl(\invol \circ \transfdet(z), \invol(z)\bigr) \\
&= \int \pushf{s}{\pi}(\rmd z') f(\invol \circ \transfdet \circ \invol (z'), z')= \iint  \pi(\rmd z') \delta_{\transfdet^{-1}(z')}(\rmd z) f(z,z') \eqsp.
\end{align}
Define
\begin{equation}
 \dlambda\bigl(\rmd (z,z')\bigr) = \dnu\bigl(\rmd(z,z')\bigr) + \dnu^{\invol}\bigl(\rmd(z,z')\bigr) =  \pi(\rmd z) \updelta_{\transfdet(z)}(\rmd z') + \pi(\rmd z') \updelta_{\transfdet^{-1}(z')}(\rmd z) \eqsp.
\end{equation}
Set $\lambda = \pi + \pushf{\transfdet^{-1}}{\pi}$ and define $k(z) = (\rmd \pi /\rmd \lambda)(z)$ .
Note that for any measurable function $f\colon \Zset^2 \to \rset_+$,
\begin{align}
  \label{eq:1}
   \int f(z,z') \dlambda\bigl(\rmd (z,z')\bigr)
   &= \int \pi(\rmd z) f\bigl(z,\transfdet(z)\bigr) + \int \pi(\rmd z') f\bigl(\transfdet^{-1}(z'),\transfdet \circ \transfdet^{-1}(z')\bigr) \\
   &= \int f\bigl(z,\transfdet(z)\bigr) \lambda(\rmd z) \eqsp.
\end{align}
Then, for any measurable function $f\colon \Zset^2 \to \rset_+$, we get since $k(z) = \rmd \pi/ \rmd \lambda (z)$,
\begin{equation}
\iint \dnu\bigl(\rmd (z,z')\bigr) f(z,z')
= \int \pi(\rmd z) f\bigl(z,\transfdet(z)\bigr) = \int k(z) f\bigl(z,\transfdet(z)\bigr) \lambda(\rmd z) \eqsp.
\end{equation}
On the other hand,
\begin{equation}
\iint \dnu\bigl(\rmd (z,z')\bigr) f(z,z')
 = \int h(z,z') f(z,z') \dlambda\bigl(\rmd (z,z')\bigr)
 = \int h\bigl(z,\transfdet(z)\bigr) f\bigl(z,\transfdet(z)\bigr) \lambda(\rmd z) \eqsp.
\end{equation}
showing that $  h\bigl(z,\transfdet(z)\bigr) = (\rmd \dnu/ \rmd \dlambda) \bigl(z, \transfdet(z)\bigr) =  k(z)$ $\lambda$-\mae~. Setting $z'= \invol(z)$ and using $\transfdet^{-1}= \invol \circ \transfdet \circ \invol$, we get
\[
h\bigl(\invol \circ \transfdet(z), \invol(z)\bigr) = h( \invol \circ \transfdet \circ \invol (z'), z')= h(\transfdet^{-1}(z'),z')
\]
Setting now $z''= \transfdet^{-1}(z')$, \ie~$z''= \transfdet^{-1}\circ \invol (z) = \invol\circ\transfdet (z)$, we finally obtain
%using $\transfdet^{-1} \circ \invol= \invol \circ \transfdet$
\[
h\bigl(\invol \circ \transfdet(z), \invol(z)\bigr) = h\bigl(z'',\transfdet(z'')\bigr) = k(z'') = k\bigl(\transfdet^{-1} \circ \invol(z)\bigr) = k\bigl(\invol \circ \transfdet(z)\bigr) \eqsp.
\]
The proof of  \eqref{eq:definition-k-lambda} and \eqref{eq:generalized-MH-deterministic} is concluded using \Cref{theo:extension-tierney}.

\subsection{Proof of \eqref{eq:generalized-MH-deterministic-dens}}
\label{sec:proof:generalized-MH-deterministic-dens}
We now consider the case $\Zset= \rset^d$ and $\pi(\rmd z)= \pi(z) \rmd z$. We first identify the dominating measure $\lambda$ defined in~\eqref{eq:definition-k-lambda}. For any nonnegative measurable function $f$,
\begin{align}
\lambda(f)
&= \int f(z) \pi(z) \rmd z + \int f \circ \transfdet^{-1}(z) \pi(z) \rmd z \\
&= \int f(z) \pi(z) \rmd z + \int f(z) \pi \circ \transfdet(z) \Jac_{\transfdet}(z) \rmd z \eqsp.
\end{align}
Hence, $\lambda(\rmd z) = \lambda(z) \rmd z$ with
\begin{equation}
\label{eq:definition-density-k}
\lambda(z)= \pi(z)+ \pi \circ \transfdet(z) \Jac_{\transfdet}(z) \eqsp.
\end{equation}
Plugging this expression in \eqref{eq:definition-k-lambda}, we get that
\begin{equation}
\label{eq:definition-k-dens}
k(z)= \frac{\rmd \pi }{\rmd  \lambda}(z)= \frac{\pi(z)}{\pi(z) + \pi \circ \transfdet(z) \Jac_{\transfdet}(z)} \eqsp.
\end{equation}
We have for any function nonnegative measurable function $f$,
\begin{align}
\label{eq:invariance-pi-invol}
\pushf{\invol}{\pi}(f)
= \int \pi(z) f \circ \invol (z) \rmd z = \int \pi \circ \invol(z) \Jac_{\invol}(z) f(z)\rmd z,
\end{align}
which implies since $\pushf{s}{\pi}= \pi$, that $\pi \circ \invol(z)= \pi(z) / \Jac_{\invol}(z)$ $\Leb_d$-\mae.
Hence, we get that
\begin{align}
\nonumber
k\bigl(\invol \circ \transfdet(z)\bigr) &= \frac{\pi\bigl(\invol \circ \transfdet(z)\bigr)}{\pi\bigl(\invol \circ \transfdet(z)\bigr) + \pi\bigl(\transfdet \circ \invol \circ \transfdet (z)\bigr) \Jac_{\transfdet}\bigl(\invol \circ \transfdet(z)\bigr)} \\
\label{eq:expression-k-invol}
&= \frac{\pi \circ \transfdet(z)}{\pi \circ \transfdet(z) + \pi(z) \rho_\transfdet(z)},
\end{align}
where we have set
\begin{equation}
\label{eq:definition-rho}
\rho_\transfdet(z)= \frac{\Jac_{\transfdet}\bigl(\invol \circ \transfdet(z)\bigr) \Jac_{\invol}\bigl(\transfdet(z)\bigr)}{\Jac_{\invol}(z)} \eqsp.
\end{equation}
Since $\transfdet \circ \invol \circ \transfdet (z) = \invol(z)$, we get that
\[
\Jac_\transfdet\bigl(\invol \circ \transfdet(z)\bigr) \Jac_{\invol}\bigl(\transfdet(z)\bigr) \Jac_{\transfdet}(z)= \Jac_{\invol}(z),
\]
which implies that
\[
\rho_\transfdet(z)= 1 / \Jac_{\transfdet}(z) \eqsp.
\]
Plugging this expression into~\eqref{eq:expression-k-invol}, we finally get that
\[
k\bigl(\invol \circ \transfdet(z)\bigr) = \frac{\pi \circ \transfdet(z) \Jac_{\transfdet}(z)}{\pi(z) + \pi \circ \transfdet(z) \Jac_{\transfdet}(z)} \eqsp.
\]
Combining this result with~\eqref{eq:generalized-MH-deterministic} and~\eqref{eq:definition-k-dens} concludes the proof of~\eqref{eq:generalized-MH-deterministic-dens}.

\section{Proofs of \Cref{sec:applications-examples}}
\subsection{Generalized Hamiltonian Monte Carlo algorithms}
\label{subsec:NICE_proofs}
Consider the two following assumptions:
\begin{assumptionN}
\label{assum:M_N_symmetric}
For any $i\in\{1,\dots,m\}$, $\transfn_{m+1-i}=\transfm_i$.
\end{assumptionN}
\begin{assumptionN}
\label{assum:h_lipschitz_constant}
For any $i \in  \{1,\dots, m\}$, $M_i$ and $N_i$ are $\lipschitz$-Lipschitz and  $h \leq c_0/[\mathrm{L}^{1 / 2}m]$, where $c_0 \approx 0.3$.
\end{assumptionN}
\begin{lemma}
Assume \Cref{assum:M_N_symmetric}. Then, $\invol\circ\transfdet\circ\invol = \transfdet^{-1}$.\footnote{This condition is missing in the main text due to a late error with our versioning system.}
\end{lemma}
\begin{proof}
Denote for $i\in\{1,\dots,m\}$, $\LF_i = \Psi_{\transfn_i}\circ\Upsilon\circ\Psi_{\transfm_i}$, where $\Upsilon(x,p) = (x+hp, p)$, $\Psi_{\transfm}(x,p) = \bigl(x, p+h\transfm(x)\bigr)$ and $\Psi_{\transfn}(x,p) = \bigl(x, p+h\transfn(x)\bigr)$.
Each of those transforms verify
\begin{equation}
\invol\circ\Upsilon\circ\invol = \Upsilon^{-1}\eqsp,\eqsp \invol\circ\Psi_{\transfm}\circ\invol = \Psi_{\transfm}^{-1}\eqsp,\eqsp \invol\circ\Psi_{\transfn}\circ\invol = \Psi_{\transfn}^{-1}\eqsp.
\end{equation}
Then, $\invol\circ\LF_i\circ\invol = \Psi_{\transfn_i}^{-1}\circ\Upsilon^{-1}\circ\Psi_{\transfm_i}^{-1}$ and thus,
\begin{equation}
\invol\circ\transfdet\circ\invol = \Psi_{\transfn_m}^{-1}\circ\Upsilon^{-1}\circ\Psi_{\transfm_m}^{-1}\circ\dots\circ\Psi_{\transfn_1}^{-1}\circ\Upsilon^{-1}\circ\Psi_{\transfm_1}^{-1}\eqsp.
\end{equation}
On the other hand,
\begin{equation}
\transfdet^{-1} = \LF_1^{-1}\circ\dots\circ\LF_m^{-1}= \Psi_{\transfm_1}^{-1}\circ\Upsilon^{-1}\circ\Psi_{\transfn_1}^{-1}\circ\dots\circ\Psi_{\transfm_m}^{-1}\circ\Upsilon^{-1}\circ\Psi_{\transfn_m}^{-1}\eqsp.
\end{equation}
Applying \Cref{assum:M_N_symmetric} concludes the proof.
\end{proof}

\subsubsection{Reversibility vs. persistency}\footnote{We follow the order of the main text here, but this discussion should be after the proofs of the following subsections.} \label{subsubsec:onpersistency}
In Subsection~\ref{sec:gener-hamilt-dynam} we define the deterministic Markov kernel on $\Zset$,
\begin{equation}
  P\big((x,p);\rmd (x',p')\big) = \bar{\alpha}(x,p) \updelta_{\Phi(x,p)}\bigl(\rmd (x',p')\bigr) + \bigl(1-\bar{\alpha}(x,p)\bigr) \updelta_{(x,-p)}\bigl(\rmd (x',p')\bigr) \eqsp,
\end{equation}
where $\bar{\alpha}(x,p)$ is given by~\eqref{eq:generalized-MH-deterministic-nice}. Such kernels are most likely not
ergodic and the momentum must be refreshed in order to lead to an ergodic Markov chain $\big(Z_i = (X_i,P_i)\big)_{i \in \nset}$. We focus on ``full refreshment", that is the scenario where the momentum is drawn afresh from its stationary distribution before applying $P$, in which case it can be checked that $(X_i)_{i \in \nset}$ is a Markov chain of (marginal) Markov kernel,
\begin{equation}
\label{eq:NICE_kernel}
  K(x,\rmd y)= K_\alpha(x,\rmd y) + \{1- K_{\alpha}(x,\rset^d) \} \updelta_{x}(\rmd y) \eqsp,
\end{equation}
where  $K_\alpha(x,\rmd y)= \int \bar{\alpha}(x,p) \varphi(p) \updelta_{G_x(p)}(\rmd y) \rmd p$ with
$G_x(p)= \projq \circ \transfdet(x,p)$, $\projq(x,p)= x$. Sampling from $K$ is described in Algorithm~\ref{alg:NICE_full_refresh}. It is also the case that $(X_i)_{i \in \nset}$ is time-reversible (see e.g.~\cite{duane:1987}), which has the disadvantage of loosing the potentially advantageous persistency features of $P$. It is possible to recover persistency by considering the mixture of kernels
\begin{equation}
\label{eq:partial_nice}
  T:=\partialrefresh P + (1-\partialrefresh) L
\end{equation}
for $\partialrefresh \in [0,1]$, where $L\bigl((x,p),\rmd (y,q)\bigr) = K(x, \rmd y) \varphi(q) \rmd q$.
The kernel $L$ refreshes independently the position $x$ and the momentum $p$. Since the target distribution is the product of $\pi_0$ and $\varphi$ (thus the position and the momentum are independent), it is easily checked as well that $L$ leaves $\pi$-invariant since
\begin{equation}
\pi L\bigl(\rmd (y,q)\bigr)=\int \pi_0(\rmd x) K(x, \rmd y) \varphi(p)\rmd p \varphi(q)\rmd q = \pi_0(\rmd y) \varphi(q) \rmd q = \pi\bigl(\rmd (y,q)\bigr)\eqsp.
\end{equation}
Note further that $L$ is $\pi$-reversible as $K$ is $\pi_0$-reversible.
 In what follows we establish $\pi_0-$irreducibility of $K$, and in particular that Proposition~\ref{prop:one-step-irreducibility} holds, which immediately implies $\pi_0\otimes (\varphi\times\Leb)-$irreducibility of $T$ and allows us to apply Theorem~\ref{theo:tierney-supp} and conclude about convergence. Sampling from $T$ is described in Algorithm~\ref{alg:NICE_randomized}. In future work we will consider the scenario where $p$ is updated using partial refreshment such as suggested in~\cite{horowitz:1991}, for example by using an AR(1) process when $\varphi$ is a normal distribution, which requires an extension of our results; see Algorithm~\ref{alg:NICE_persistent}.

\subsubsection{Proof of \eqref{eq:acceptance-NICE} and \Cref{theo:condition-diffeo}}
\label{sec:proof:acceptance-NICE}
\begin{comment}
\subsubsection{Proof of ergodicity}
\label{SPsec:ergodicity}
We first expose two theorems showing conditions for ergodicity and thus convergence of MCMC with NICE proposals.
\begin{theorem}
\label{theorem:diffeo_ergodic}
Let $K_\alpha$ be a GMH, written as $K_\alpha(x,\rmd y)= \int \bar{\alpha}(x,p) \varphi(p) \updelta_{G_x(p)}(\rmd y) \rmd p$. If $G_x: p \mapsto \projq \circ \transfdet(x,p)$ is a $\rmc^1$-diffeomorphism with $J_{G_x}$ is not null a.e. on $\rset^d$, and $\varphi$ is strictly positive on $\rset^d$ then $P$ is $\pi$-irreducible.
\end{theorem}
\begin{proof}
If $G_x$ is a $\rmc^1$-diffeomorphism, then we can write for any measurable function $f$,
\begin{equation}
K_\alpha f(x)= \int \bar{\alpha}(x,p) \varphi(p) f(G_x(p)) \rmd p= \int \bar{\alpha}(x,G_x^{-1}(y)) \varphi(G_x^{-1}(y)) f(y) \Jac_{G_x^{-1}}(y)\rmd y \eqsp,
\end{equation}
and $K_\alpha$ has a density \wrt~the Lebesgue measure as we can write
$K_\alpha(x,\rmd y)= \bar{\alpha}(x,G_x^{-1}(y)) \varphi(G_x^{-1}(y)) \Jac_{G_x^{-1}}(y)\rmd y$.

Then, if $y$ is such that $\pi(y)>0$, then $p(x,y)>0$ and $p(s(x), s(y))>0$ for any $x\in\rset^d$.
Then, by \Cref{theo:irred_and_co}, $K_\alpha$ is $\pi$-irreducible.
\end{proof}
\end{comment}
We first establish the elementary equation \eqref{eq:acceptance-NICE}.
\begin{lemma}
Assume that for each $x\in\rset^d$, $G_x:p\mapsto\projq\circ\transfdet(x,p)$ is a $\rmc^1$-diffeomorphism. Then, $K_\alpha(x,\rmd y)$ has a density $K_\alpha(x,\rmd y) = \alpha(x,y)q(x,y)\rmd y$ where
\begin{align}
\alpha(x,y) &= \accfun\left(\frac{\target(y) \varphi\bigl\{H_x\bigl(G_x^{-1}(y)\bigr)\bigr\}}{\target(x) \varphi\bigl(G_x^{-1}(y)\bigr)} \right)\eqsp,\\
q(x,y) &= \varphi(G_x^{-1}(y))\Jac_{G_x^{-1}}(y)\eqsp.
\end{align}
\end{lemma}
\begin{proof}
First by~\eqref{eq:generalized-MH-deterministic-nice}, for any $(x,p) \in \rset^{2d}$, we have by definition
%\christophe{Acceptance function missing below...}
\begin{equation}
  K_{\alpha}f(x) = \int \bar{\alpha}(x,p) \varphi(p) f\bigl(G_x(p)\bigr) \rmd p = \int \accfun\left(\frac{\pi \circ \Phi (x,p)}{\pi(x,p)}\right) \varphi(p) f\bigl(G_x(p)\bigr) \rmd p \eqsp.
\end{equation}
Then, using the change of variable $y = G_x(p)$, we obtain
\begin{equation}
  K_{\alpha}f(x) = \int \accfun\left(\frac{\pi \circ \Phi\bigl(x,G_x^{-1}(y)\bigr)}{\pi\bigl(x,G_x^{-1}(y)\bigr)}\right) \varphi\bigl(G_x^{-1}(y)\bigr) \Jac_{G^{-1}_x}(y) f(y) \rmd y \eqsp,
\end{equation}
which concludes the proof of~\eqref{eq:acceptance-NICE} since $\pi = \pi_0 \otimes \varphi$.
\end{proof}

We now prove \Cref{theo:condition-diffeo} which gives conditions on the mappings $\{\transfm_i, \transfn_i\}_{i= 1}^m$ that ensure that for all $x\in\rset^d$, $G_x$ is a $\rmc^1$-diffeomorphism.
\begin{theorem}
\label{theo:lip_nice}
%Assume that the transformations $\{\transfm_k,\transfn_k\}_{k=1}^m$ are continuously differentiable on $\rset^d$ and $\lipschitz$-Lipschitz.
%Assume in addition that there exists a constant $\lipschitzM < \infty$, $\| \transfm_k(x) \| \leq \lipschitzM (1 + \|x \|) $ and
%$\| \transfn_k(x) \| \leq \lipschitzM (1 + \|x \|) $.
%for all $x \in \rset^d$
%Set  $h < h_0(\lipschitz, m) = c_0\lipschitz^{-1 / 2}/m$ where $c_0= 0.29$.
Assume \Cref{assum:h_lipschitz_constant}. Then, for any $x\in \rset^d$, the function $G_x(p)= \projq \circ \transfdet(x,p)$ is a $\rmc^1$ diffeomorphism.
Moreover, the GMH kernel based on NICE transitions is ergodic.
\end{theorem}
We preface the proof by some auxiliary results.
Recall that one step of the NICE transition is given by $\LF_{i}(x_i,p_i) = (x_{i+1}, p_{i+1})$, where:
\begin{equation}
\begin{cases} p_{i+1/2} &=  p_i + h \transfm_i(x_i),
\\
x_{i+1} &= x_i + h p_{i+ 1/2},
\\
p_{i+1} &= p_{i+1/2} + h N_{i}(x_{i+1}).
\end{cases}
\end{equation}
Denote
\begin{equation}
\Lambda^{(j)}= \LF_{j} \circ \dots \circ \LF_1 \eqsp.
\end{equation}
\begin{lemma}
For all  $k\in\nsets$, we get
\begin{align}
&x_k = x_1 + (k-1)hp_1 + h^2\sum_{i=1}^{k-1}(k-i)\transfm_i(x_i) + h^2\sum_{i=1}^{k-2}(k-1-i)\transfn_i(x_{i+1}) \eqsp,
\\
&p_k = p_1 + h\sum_{i=1}^{k-1} \transfm_i(x_i)+ h \sum_{i=1}^{k-1} \transfn_i(x_{i+1}) \eqsp.
\end{align}
\end{lemma}
\begin{proof}
The proof proceeds by induction.  The assertion is obviously true for $k=2$.
Let us suppose that the assertion holds true for some $k\in\nsets$.
\begin{align}
p_{k+1/2} &=  p_k + h \transfm_k(x_k)=p_1 +h \sum_{i=1}^{k} \transfm_i(x_i)+ h \sum_{i=1}^{k-1} \transfn_i(x_{i+1}),\\
x_{k+1} &= x_k + hp_{k+1/2}\\
&=  x_1 + (k-1)hp_1 + h^2 \sum_{i=1}^{k-1}(k-i)\transfm_i(x_i) + h^2\sum_{i=1}^{k-2}(k-1-i)\transfn_i(x_{i+1})\\
& \quad \quad +h\left(p_1 +h \sum_{i=1}^{k} \transfm_i(x_i)+ h \sum_{i=1}^{k-1} \transfn_i(x_{i+1})\right)\\
&=x_1 + khp_1 + h^2 \sum_{i=1}^{k}(k+1-i)\transfm_i(x_i) + h^2 \sum_{i=1}^{k-1}(k-i)\transfn_i(x_{i+1}), \\
p_{k+1} &= p_{k+1/2} + h N_{k}(x_{k+1})=p_1 +h \sum_{i=1}^{k} \transfm_i(x_i)+ h \sum_{i=1}^{k} \transfn_i(x_{i+1}).
\end{align}
This concludes the proof.
\end{proof}
Denote  for all $(x_1,p_1)\in\rset^{2d}$,
\begin{align}
&G_{x_1}(p_1) = x_1 + mhp_1  + h^2 \Theta_m(x_1,p_1), \\
&\Theta_m(x_1,p_1) = \sum_{i=1}^{m}(m+1-i)\transfm_i(x_i) + \sum_{i=1}^{m-1}(m-i)\transfn_i(x_{i+1}).
\end{align}

Since the mappings $\{M_k\}_{k=1}^m$, $\{N_k\}_{k=1}^m$ are continuously differentiable, the mapping $\Xi_k$ is continuously differentiable.

% and we can write
%\begin{equation}
%\Jac_{p, \Xi_k}(x_0,p_0) =\sum_{i=0}^{k-1}(k-i)\Jac_{p, \transfm_i}\circ\tilde{\Phi}^{(i)}(x_0, p_0) \Jac_{p, \tilde{\Phi}^{(i)}}(x_0, p_0)+ \sum_{i=0}^{k-2}(k-i)\Jac_{p, \transfn_i}\circ\tilde{\Phi}^{(i+1)}(x_0, p_0) \Jac_{p, \tilde{\Phi}^{(i+1)}}(x_0, p_0)\eqsp.
%\end{equation}
\begin{lemma}
\label{lemma:s1}
For $(x_1, p_1), (\tilde{x}_1, \tilde{p}_1) \in\rset^{2d}$, denote $(x_{k+1}, p_{k+1}),(\tilde{x}_{k+1}, \tilde{p}_{k+1})$ the states obtained after $k$ NICE-based transitions.
Under the Lipschitz constraint $\lipschitz$, we have
\begin{equation}
\| x_{k+1}-\tilde{x}_{k+1}\| + \lipschitz^{-1/2}\|p_{k+1}-\tilde{p}_{k+1}\| \leq
\left\{1+h \lipschitz^{1 / 2} \vartheta_{1}\left(h \lipschitz^{1 / 2}\right)\right\}^{k}\left\{\left\|x_{1}-\tilde{x}_1\right\|+\lipschitz^{-1 / 2}\left\|p_{1}-\tilde{p}_1\right\|\right\}\eqsp,
\end{equation}
where $\vartheta_{1}\left(s\right) = 2 + s + s^2$.
\end{lemma}
\begin{proof}
We show this result for $k=1$ and the apply a straightforward induction.
For $k=1$, we have
\begin{align}
\left\|x_{2}-\tilde{x}_{2}\right\| &=\left\|x_{1}+h^{2} \transfm_1\left(x_{1}\right) +h p_{1}-\left\{\tilde{x}_{1}+h^{2}  \transfm_1\left(\tilde{x}_{1}\right)+h \tilde{p}_{1}\right\}\right\| \\
 & \leq\left(1+h^{2} \lipschitz \right)\left\|x_{1}-\tilde{x}_{1}\right\|+h\left\|p_{1}-\tilde{p}_{1}\right\| \eqsp.
\end{align}
Moreover, we have
\begin{align}
&\| p_{2}-\tilde{p}_{2} \| =\left\|p_{1}-\tilde{p}_{1}-h \left\{\transfn_1\left(x_{2}\right)+\transfm_1\left(x_{1}\right)\right\}+h \left\{\transfn_1\left(\tilde{x}_{2}\right)+\transfm_1\left(\tilde{x}_{1}\right)\right\}\right\| \\
&\quad \leq\left\|p_{1}-\tilde{p}_{1}\right\| + h \lipschitz \left\{\left\|\tilde{x}_{2}-x_{2}\right\|+\left\|\tilde{x}_{1}-x_{1}\right\|\right\} \\
&\quad \leq\left(1+h^{2} \lipschitz \right)\left\|p_{1}-\tilde{p}_{1}\right\|+h \lipschitz\left(2+h^{2} \lipschitz \right)\left\|x_{1}-\tilde{x}_{1}\right\| \eqsp.
\end{align}
Summing the two previous expression, we get the desired result for $k=1$.
%Finally, a direct induction allows us to write this lemma for $k\in\nsets$.
\end{proof}
\begin{lemma}
\label{lemma:s3}
For any $h>0$, we have
\begin{equation}
\sup _{(x, p, v) \in \rset^{3 d}}\left\{\left\|\Theta_m(x, p)-\Theta_m(x, v)\right\| /\|p-v\|\right\}
\leq(m / h)\left\{\left(1+h \lipschitz^{1 / 2} \vartheta_{1}\left(h \lipschitz\right)\right)^{m}-1\right\} \eqsp.
\end{equation}
\end{lemma}
\begin{proof}
By \Cref{lemma:s1}, we have that, for any $(x,p,v)\in\rset^{3d}$,
\begin{equation}
\|\projq \circ \Lambda^{(m)}(x,p)-\projq \circ \Lambda^{(m)}(x,v)\|\leq \left\{1+h \lipschitz^{1 / 2} \vartheta_{1}\left(h \lipschitz^{1 / 2}\right)\right\}^{m} \lipschitz^{-1/2} \|p-v\| \eqsp.
\end{equation}
Denote $\Lambda_1^{(i)}= \projq \circ \Lambda^{(i)}$ and as a convention $\Lambda_1^{(0)}= \projq $. We obtain
\begin{align}
&\|\Theta_m(x,p)-\Theta_{m}(x,v)\|%\leq
%\sum_{i=i}^{m}(m+1-i)\|\transfm_i(\transfdet_1^{(i-1)}(x,p))-\transfm_i(\transfdet_1^{(i-1)}(x,v))\| + \sum_{i=1}^{m-1}(m-i)\|\transfn_i(\transfdet_1^{(i-1)}(x,p))-\transfn_i(\transfdet_1^{(i-1)}(x,v))\|
\\
&\qquad \leq\lipschitz \left(\sum_{i=1}^{m-1}2(m+1-i)\|\Lambda_1^{(i-1)}(x,p)-\Lambda_1^{(i-1)}(x,v)\| + \|\Lambda_1^{(m-1)}(x,p)-\Lambda_1^{(m-1)}(x,v)\|\right)
\\
&\qquad\leq \lipschitz^{1/2} \left(\sum_{i=1}^{m-1}2(m+1-i) \left\{1+h \lipschitz^{1 / 2} \vartheta_{1}\left(h \lipschitz^{1 / 2}\right)\right\}^{i-1} +  \left\{1+h \lipschitz^{1 / 2} \vartheta_{1}\left(h \lipschitz^{1 / 2}\right)\right\}^{m-1}\right)\|p-v\|
\\
&\qquad\leq 2m \lipschitz^{1/2} \left\{\left(1+h \lipschitz^{1 / 2} \vartheta_{1}\left(h \lipschitz^{1 / 2}\right)\right)^{m}-1\right\}/\left(h\lipschitz^{1 / 2} \vartheta_{1}\left(h \lipschitz^{1 / 2}\right)\right)\|p-v\|
\\
&\qquad \leq (m/h) \left\{\left(1+h \lipschitz^{1 / 2} \vartheta_{1}\left(h \lipschitz^{1 / 2}\right)\right)^{m}-1\right\} \|p-v\| \eqsp,
\end{align}
as $ \vartheta_{1}\left(h \lipschitz^{1 / 2}\right)\geq 2$.
\end{proof}

We can now prove \Cref{theo:lip_nice}.
\begin{proof}
Note that for any $h>0$, $m\in\nsets$, we have
\begin{equation}
\left(1+h \lipschitz^{1 / 2} \vartheta_{1}\left(h \lipschitz^{1 / 2}\right)\right)^{m} -1 \leq \exp\left\{h \lipschitz^{1 / 2}m \vartheta_{1}\left(h \lipschitz^{1 / 2}m\right)\right\} -1\eqsp,
\end{equation}
as $\vartheta_1$ is non decreasing.
%We thus want first to find some $c\in\rset^+$ such that $e^{c\vartheta_1(c)} =3/2$.
The function $c\to \rme^{c\vartheta_1(c)}$ is continuous and strictly increasing, from 0 to $\infty$ on $\rset$, thus   $\rme^{c\vartheta_1(c)} =2$ admits a unique solution,
for $c_0\approx 0.29$.
For $c<c_0$, we have $\rme^{c\vartheta_1(c)} <2$. In particular, if $h < h_0(\lipschitz, k) = c_0/\lipschitz^{1 / 2}m$, then
\begin{equation}
\left\{\left(1+h \lipschitz^{1 / 2} \vartheta_{1}\left(h \lipschitz^{1 / 2} k\right)\right)^{m}-1\right\} < 1\eqsp.
\end{equation}
We first prove that, for all $x_1 \in \rset^d$,
\begin{equation}
\label{eq:one-to-one-assertion}
\text{
the function $p \mapsto p + h/m \Theta_m(x_1,p)$ is one-to-one.
}
\end{equation} By \Cref{lemma:s3}, there exists $0< \kappa < 1$ such that for all $p,v \in \rset^d$,
\[
\| H_{y_1}(p) - H_{y_1}(v)\| \leq \frac{h}{m}
\| \Theta_m(x_1,p) - \Theta_m(x_1,v) \| \leq \kappa \| p - v \| \eqsp,
\]
where $H_{y_1}\colon p\mapsto y_1 - h/m \Theta_m(x_1,p)$. Hence, by the Banach fixed point theorem, for any $y_1 \in \rset^d$, $H_{y_1}$ has a unique fixed point $p_1$ and
\[
y_1= p_1 + \frac{h}{m} \Theta_m(x_1,p_1) \,
\]
showing~\eqref{eq:one-to-one-assertion}.
Hence
\[ p \mapsto G_{x_1}(p) = x_1 + mh p + h^2 \Theta_m(x_1,p)
\]
is one-to-one. Since in addition $\Theta_m$ is continuously differentiable and the Jacobian of $G_{x_1}$ is invertible, the function
%the  the inverse function theorem
%In particular, under this condition, for any $x\in \rset^d$, the function
$G_x$ is a $\rmc^1$ diffeomorphism.
%Then, the GMH kernel based on NICE proposals is $\pi$-irreducible, and we can apply \Cref{theo:irred_and_co}.
\end{proof}

\subsubsection{Proof of~\eqref{eq:NICE-identity-ratio}}
\label{sec:proof:NICE-identity-ratio}
The result~\eqref{eq:NICE-identity-ratio} is directly linked to \Cref{theo:convergence_nice} which ensures convergence of the Markov kernel based on NICE proposals.
\begin{theorem}
\label{theo:convergence_nice}
Assume \Cref{assum:M_N_symmetric} and \Cref{assum:h_lipschitz_constant}. Then, the Markov kernel $K$ defined in~\eqref{eq:NICE_kernel} is a $\target$-reversible MH kernel with transition density 
\begin{equation}
q(x,y) = \varphi\bigl(G_x^{-1}(y)\bigr)\Jac_{G_x^{-1}}(y)\eqsp,
\end{equation}
and acceptance probability
\begin{equation}
\alpha(x,y) = \accfun\left(\frac{\target(y)q(y,x)}{\target(x)q(x,y)}\right)\eqsp.
\end{equation}
In addition, \Cref{theo:basic-cv-textbook} applies. 
\end{theorem}
\begin{proof}
Note that for all $(x,p) \in \rset^{2d}$,
\begin{equation}
\label{eq:basic-relation-transfdet}
\transfdet^{-1} \circ \transfdet(x,p)= \transfdet^{-1}\bigl(G_x(p),H_x(p)\bigr) = (x,p),
\end{equation}
where we have used $G_x(p)= \projq \circ \transfdet(x,p)$ and $H_x(p)= \projp \circ \transfdet(x,p)$.
Under \Cref{assum:h_lipschitz_constant}, for any $x \in \rset^d$, $p \mapsto G_x(p)$ is a diffeomorphism.
Then, plugging $y= G_x(p)$, $p= G_x^{-1}(y)$ in \eqref{eq:basic-relation-transfdet}, we obtain
\begin{equation}
\label{eq:basic-relation-transfdet-1}
\transfdet^{-1}\bigl(y,H_x \circ G_x^{-1}(y)\bigr)= \bigl(x, G_x^{-1}(y)\bigr) \eqsp.
\end{equation}
Under \Cref{assum:M_N_symmetric},  $\invol \circ \transfdet \circ \invol = \transfdet^{-1}$. Then,  we get
\begin{equation}
\label{eq:basic-relation-transfdet-2}
\transfdet\bigl(y,-H_x \circ G_x^{-1}(y)\bigr) = \bigl(x,-G_x^{-1}(y)\bigr) \eqsp.
\end{equation}
Hence $G_y\bigl(-H_x \circ G_x^{-1}(y)\bigr) = x$ or equivalently, $-H_x \circ G_x^{-1}(y)= G_y^{-1}(x)$. Since $\varphi$ is even, this implies
\begin{equation}
\label{eq:basic-relation-transfdet-3}
\varphi\bigl(H_x \circ G_x^{-1}(y)\bigr)= \varphi\bigl(G_y^{-1}(x)\bigr) \eqsp.
\end{equation}
Recall that $\Jac_{\transfdet}(x,p)= 1$, for all $(x,p) \in \rset^{2d}$. Using again $-H_x \circ G_x^{-1}(y)= G_y^{-1}(x)$ in~\eqref{eq:basic-relation-transfdet-3}, we get
\begin{equation}
\transfdet\bigl(y,G_y^{-1}(x)\bigr) = \bigl(x,-G_x^{-1}(y)\bigr) \eqsp.
\end{equation}
Using the chain rule for Jacobian matrices, we get
\begin{equation}
\label{eq:chain-rule-jacobian}
\Jac_{G_y^{-1}}(x)= \Jac_{G_x^{-1}}(y).
\end{equation}
Combining~\eqref{eq:basic-relation-transfdet-3} and~\eqref{eq:chain-rule-jacobian} leads to \eqref{eq:NICE-identity-ratio} by noting that
\begin{equation}
\label{eq:NICE-identity-ratio-sup}
\frac{q(y,x)}{q(x,y)}= \frac{\varphi\bigl(H_x \circ G_x^{-1}(y)\bigr)}{\varphi\bigl(G_x^{-1}(y)\bigr)} \eqsp.
\end{equation}
Hence, the acceptance ratio $\alpha$ coincides with the standard MH ratio and the marginal Markov kernel $K$ is thus $\pi_0$-reversible. We also note that $q(x,y)$ satisfies the conditions of \Cref{theo:irred_and_co} given the assumptions on $\varphi$ and $G_x$.
Moreover, $K$ is $\target$-irreducible, by \Cref{theo:condition-diffeo}.
Then, \Cref{theo:basic-cv-textbook} applies.
\end{proof}

\subsubsection{Implementation details}
%Then, for all $i\in\{1,\dots,m\}$, $\LF_i(x_{m+2-i}, -p_{m+2 - i}) = (\tilde{x}, \tilde{p})$ where 
%\begin{align}
%\tilde{p}_{1/2} &= -p_{m+2-i} + h\transfm_i(x_{m+2-i})\\
%&= -p_{m+1 +1/2-i} - h\transfn_{m+1-i}(x_{m+2 -i}) + h\transfm_i(x_{m+2-i})\\
%&= -p_{m+1 +1/2-i}\eqsp,\\
%\tilde{x} &= x_{m+2-i} + h\tilde{p}_{1/2}\\
%&= x_{m+1 -i } + hp_{m+1 +1/2-i} - hp_{m+1 +1/2-i} =   x_{m+1 -i }\eqsp,\\
%\tilde{p} &= -p_{m+1 +1/2-i} + h\transfn_i(x_{m+1 -i })\\
%&= - p_{m+1 -i } - h\transfm_{m+1 -i}(x_{m+1 -i }) +h\transfn_i(x_{m+1 -i })\\
%&= -p_{m+1 -i }\eqsp,
%\end{align}
%and for all $i\in\{1,\dots,m\}$, $\LF_i(x_{m+1 - (i-1)}, -p_{m+1 - (i-1)}) = (x_{m+1 -i}, -p_{m+1 -i})$, which imply by an immediate induction that $\transfdet\circ\invol(x_{m+1},p_{m+1}) = (x_1, -p_1)$ and thus $\transfdet^{-1} = \invol\circ\transfdet\circ\invol$.
Algorithm~\ref{alg:NICE_full_refresh} presents the methodology for sampling according to the kernel $K$~\eqref{eq:marginal-NICE-kernel}, which is $\target$-reversible.
%At each step of the algorithm, a new value of the auxiliary momentum $p$ is sampled. This is intuitively linked to the fact that we can write $\pi = \target\otimes\varphi$. Then, refreshing fully the momentum leaves the joint factorized distribution invariant
\begin{algorithm}[!ht]
    \caption{NICE with full refreshment at each iteration}
    \label{alg:NICE_full_refresh}
    \begin{algorithmic}
      \STATE {\bfseries Input:} Transformation $\transfdet$ and momentum-flip involution $\invol$, acceptance function $\accfun$, unnormalized target density $\pi$, density $\varphi$ of momentum $p$, initial point $x_0$, number of steps $N$
      \FOR{$i=0$ {\bfseries to} $N-1$}
        \STATE Draw $q_{i}\sim \varphi$;
        \STATE Compute proposal $(y_{i+1}, q_{i+1}) =\transfdet(x_i, q_i)$;
        \STATE Draw $B_i\sim\Ber(a_i)$ where
        \[a_i = \accfun\left( \frac{\target(y_{i+1})\varphi(q_{i+1})}{\target(x_i)\varphi(q_i)}\right)\eqsp;
        \]
        \IF{$B_i\equiv 1$}
          \STATE Set $x_{i+1}= y_{i+1}$;
        \ELSE
          \STATE Set $x_{i+1}= x_{i}$;
        \ENDIF
      \ENDFOR
      \STATE Return $(x_{0:N})$
    \end{algorithmic}
  \end{algorithm}
\newpage
In order to recover persistency, as discussed in \Cref{subsubsec:onpersistency}, we consider the mixture of kernels $T$ \eqref{eq:partial_nice}; see Algorithm~\ref{alg:NICE_randomized}. 
%Instead of refreshing fully at each step the auxiliary momentum, we know set a hyperparameter $\partialrefresh$ as a probability of refreshing the momentum.
%This allows to persistently propose a move
%We also present another kind of consistency, described for example in \citep{horowitz:1991}.
%In this work, at each step, th auxiliary momentum variable $p$ is refreshed through someMarkov kernel leaving $\varphi$ invariant. In particular, when $\varphi$ is a standard normal Gaussian, we perform a partial momentum refreshment, setting $p\leftarrow \partialrefresh p + \sqrt{1-\partialrefresh^2} u$ where $u\sim\varphi$. Again, the hyperparameter $\partialrefresh$ controls the ``amount of persistency'' of the algorithm.

\begin{algorithm}[!ht]
    \caption{NICE with randomized full refreshment}
  \label{alg:NICE_randomized}
    \begin{algorithmic}
      \STATE {\bfseries Input:} Transformation $\transfdet$ and momentum-flip involution $\invol$, acceptance function $\accfun$, unnormalized target $\pi$, density $\varphi$ of momentum $p$, probability of refreshment $\partialrefresh$, initial point $x_0$ and initial momentum $p_0$, number of steps $N$;
      \FOR{$i=0$ {\bfseries to} $N-1$}
      \STATE Draw $R_i\sim\Ber(\partialrefresh)$;
      \IF{$R_i \equiv 0$}
        \STATE Compute proposal $(y_{i+1}, q_{i+1}) =\transfdet(x_i, p_i)$; \hfill\emph{\#\#\# No refreshment, deterministic dynamics}
        \STATE Draw $B_i\sim\Ber(a_i)$ where
        \[
          a_i = \accfun\left( \frac{\target(y_{i+1})\varphi(q_{i+1})}{\target(x_i)\varphi(q_i)}\right)\eqsp;
        \]
        \IF{$B_i \equiv 1$}
          \STATE Set $(x_{i+1},p_{i+1})= (y_{i+1}, q_{i+1})$;\hfill\emph{\#\#\# accept the move and keep the momentum}
        \ELSE
        \STATE Set $(x_{i+1}, p_{i+1})= \invol(x_{i}, p_i)$;\hfill{ \emph{\#\#\# reject the move and flip the momentum} }
        \ENDIF
      \ELSE
        \STATE Sample $q_{i}\sim \varphi$; \hfill\emph{\#\#\# Full refreshment of the momentum to update the position}
        \STATE Compute proposal $(y_{i+1}, q_{i+1}) =\transfdet(x_i, q_i)$;
        \STATE Draw $B_i\sim\Ber(a_i)$ where
        \[a_i = \accfun\left( \frac{\target(y_{i+1})\varphi(q_{i+1})}{\target(x_i)\varphi(q_i)}\right)\eqsp;
        \]
         \STATE Draw $p_{i+1}\sim\varphi$;
        \IF{$B_i\equiv 1$}
          \STATE Set $x_{i+1} = y_{i+1}$;
        \ELSE
          \STATE Set $x_{i+1}= x_{i}$;
        \ENDIF
      \ENDIF
      \ENDFOR
      \STATE Return $(x_{0:N})$
    \end{algorithmic}
  \end{algorithm}
\newpage

  \begin{algorithm}[!ht]
    \caption{NICE with persistence}
    \label{alg:NICE_persistent}
    \begin{algorithmic}
      \STATE {\bfseries Input:} Transformation $\transfdet$ and momentum-flip involution $\invol$, acceptance function $\accfun$, unnormalized target $\pi$, density $\varphi$ of momentum $p$, hyperparameter $\beta$, initial point $x_0$ and initial momentum $p_0$, number of steps $N$
      \FOR{$i=0$ {\bfseries to} $N-1$}
        \STATE Draw $u_{i}\sim \varphi$ and set $q_i = \beta p_i + \sqrt{1-\beta^2} u_i$;
        \STATE Compute proposal $(y_{i+1}, q_{i+1}) =\transfdet(x_i, q_i)$;
        \STATE Draw $B_i\sim\Ber(a_i)$ where
        \[a_i = \accfun\left( \frac{\target(y_{i+1})\varphi(q_{i+1})}{\target(x_i)\varphi(q_i)}\right)\eqsp;
        \]
        \IF{$B_i\equiv 1$}
          \STATE Set $(x_{i+1},p_{i+1})= (y_{i+1}, q_{i+1})$;\hfill\emph{\#\#\# accept the move and keep the momentum}
        \ELSE
          \STATE Set $(x_{i+1}, p_{i+1})= \invol(x_{i}, q_i)$;\hfill{ \emph{\#\#\#reject the move and flip the momentum} }
        \ENDIF
      \ENDFOR
      \STATE Return $(x_{0:N})$
    \end{algorithmic}
  \end{algorithm}

\subsection{Proof of \eqref{eq:alpha_lifted}}
\label{sec:proof:alpha_lifted}
By~\eqref{eq:lifted_density}, we get
\begin{align}
q\bigl((x,v), (y, w)\bigr) &=  \{\probav \indi{v}(w) +(1-\probav) \indi{-v}(w)\}  q_w(x,y) \eqsp, \\
q\bigl(\invol(y,w), \invol(x,v)\bigr) &= q\bigl((y, -w) ,  (x,-v)\bigr) \\
&= \{\probav \indi{-w}(-v) +(1-\probav) \indi{w}(-v)\} q_{-w}(y,x) \\
&= \{\probav \indi{v}(w) +(1-\probav) \indi{-v}(w)\} q_{-w}(y,x) \eqsp,
\end{align}
which implies that
\begin{equation}
\frac{q\bigl(\invol(y,w), \invol(x,v)\bigr)}{q\bigl((x,v), (y, w)\bigr)}= \frac{q_{-w}(y,x)}{q_w(x,y)} \eqsp.
\end{equation}
The proof follows from~\eqref{eq:alpha_density}.

\subsection{Implementation details of \Cref{example:cIT}}
\label{subsec:cIT_implementation}

We define here a probability of refresh $\partialrefresh$.
At each iteration, we refresh the direction with probability $\partialrefresh$, in which case we draw $v\sim\Uniform\{-1,1\}$. With this definition, we can reinterpret the parameter $\probav$ \eqref{eq:lifted_density}
{\small\begin{equation}
\label{eq:lifted_density}
q\bigl((x,v), (y, w)\bigr)= \bigl\{\probav \indi{v}(w) +(1-\probav) \indi{-v}(w)\bigr\}  q_w(x,y)\eqsp,
\end{equation}}
as $\partialrefresh = 2\probav$.
In particular, we can write the lifted algorithm with randomized direction refresh in Algorithm~\ref{alg:lifted Markov Sampling}.
\begin{algorithm}[!ht]
    \caption{Lifted  Markov sampling}
    \label{alg:lifted Markov Sampling}
	\begin{algorithmic}
      \STATE {\bfseries Input:} Transformations $G_{1,x}, G_{-1,x}$, acceptance function $\accfun$, unnormalized target $\pi$, density $\varphi$ of momentum $p$, initial point $x_0$ and initial direction $v_0$, probability of refreshment $\partialrefresh$, number of steps $N$
      \FOR{$i=0$ {\bfseries to} $N-1$}
        \STATE Draw $R_i\sim\Ber(\partialrefresh)$;
        \IF{$R_i\equiv 1$}
        \STATE Refresh direction $w_i \sim \Uniform\{-1,1\}$;
        \ELSE
        \STATE Keep direction $w_i=v_i$;
        \ENDIF
        \STATE Draw $q_i\sim \varphi$;
         \STATE Compute proposal $y_{i+1} =G_{w_i, x_i}(q_i)$;
        \STATE Draw $B_i\sim\Ber(a_i)$ where
        \[a_i = \accfun\left( \frac{\target(y_{i+1})\varphi\bigl({G_{-w_i,y_{i+1}}^{-1}}(x_i)\bigr) \Jac_{G_{-w_i,y_{i+1}}^{-1}}(x_i) }{\target(x_i)\varphi\bigl(G_{w_i, x_i}^{-1}(y_{i+1})\bigr)\Jac_{G_{w_i,x_{i}}^{-1}}(y_{i+1})}\right) \eqsp;
        \]
        \IF{$B_i\equiv 1$}
          \STATE Set $(x_{i+1},v_{i+1})= (y_{i+1}, w_i)$;\hfill\emph{\#\#\# accept the move and keep the direction}
        \ELSE
          \STATE Set $(x_{i+1}, v_{i+1})= (x_{i}, - w_i)$;\hfill{ \emph{\#\#\#reject the move and flip the direction} }
        \ENDIF
      \ENDFOR
      \STATE Return $(x_{0:N})$
    \end{algorithmic}
  \end{algorithm}

\subsection{Proof of \Cref{lem:altern_accfun}}
\label{sec:proof:altern_accfun}
From~\eqref{eq:alpha_lifted} and~\eqref{eq:def_q_v}, we get
\begin{align}
  \alpha\bigl((x,v),(y,w)\bigr)
  &= \accfun\left(\frac{q_{-w}(y,x) \target(y)}{q_w(x,y)\target(x)}\right) \\
  \label{eq:key-identity-altern}
  &= \accfun\left(\frac{\target(y) \varphi\bigl(\tG_{-w,y}^{-1}(x)\bigr) \Jac_{\tG_{-w,y}^{-1}}(x)}
  {\target(x)  \varphi\bigl(\tG_{w,x}^{-1}(y)\bigr) \Jac_{\tG_{w,x}^{-1}}(y)} \right) = \accfun\left(\frac{\mu\bigl(y,\tG_{-w,y}^{-1}(x)\bigr) \Jac_{\tG_{-w,y}^{-1}}(x)}
  {\mu\bigl(x,\tG_{w,x}^{-1}(y)\bigr) \Jac_{\tG_{w,x}^{-1}}(y)} \right)\eqsp.
\end{align}
Set  $\tH_{w,x}(p)= \projp \circ \transflif^w(x,p)$. Note that
\[
  \transflif^w(x,p)= \bigl(\tG_{w,x}(p), \tH_{w,x}(p)\bigr) \eqsp.
\]
Hence, we obtain
\[
  \transflif^{-w}\bigl(y,\tG_{-w,y}^{-1}(x)\bigr) = \bigl(\tG_{-w,y} \circ \tG_{-w,y}^{-1}(x), \tH_{-w,y} \circ \tG_{-w,y}^{-1}(x)\bigr)
  = \bigl(x, \tH_{-w,y} \circ \tG_{-w,y}^{-1}(x)\bigr) \eqsp,
\]
which implies,
\begin{align}
\label{eq:first-identity-altern}
\bigl(y,\tG_{-w,y}^{-1}(x)\bigr)
&= \transflif^{w}\bigl(x, \tH_{-w,y} \circ \tG_{-w,y}^{-1}(x)\bigr) \eqsp, \\
&= \bigl(\tG_{w,x} \circ  \tH_{-w,y} \circ \tG_{-w,y}^{-1}(x), \tH_{w,x} \circ  \tH_{-w,y} \circ \tG_{-w,y}^{-1}(x)\bigr) \eqsp.
\end{align}
This identity in particular shows that $y= \tG_{w,x} \circ \tH_{-w,y} \circ \tG_{-w,y}^{-1}(x)$ or equivalently $\tG_{w,x}^{-1}(y)= \tH_{-w,y} \circ \tG_{-w,y}^{-1}(x)$, which used in~\eqref{eq:first-identity-altern} establishes
\begin{equation}
\label{eq:second-identity-altern}
\bigl(y,\tG_{-w,y}^{-1}(x)\bigr)= \transflif^{w}\bigl(x, \tG_{w,x}^{-1}(y)\bigr) \eqsp.
\end{equation}
Set $A_w(x,y) = \bigl(y,\tG_{-w,y}^{-1}(x)\bigr)$ and $B_w(x,y)= \bigl(x,\tG_{w,x}^{-1}(y)\bigr)$. Note that
$\Jac_{A_w}(x,y)= \Jac_{\tG_{-w,y}^{-1}}(x)$, $\Jac_{B_w}(x,y)= \Jac_{\tG_{w,x}^{-1}}(y)$ and by~\eqref{eq:second-identity-altern} and the chain rule
\[
\Jac_{A_w}(x,y)= \Jac_{\transflif^w}\bigl(B_w(x,y)\bigr) \Jac_{B_w}(x,y) \eqsp,
\]
which implies
\begin{equation}
\label{eq:relation-jacobian-altern}
\Jac_{\transflif^{w}}\bigl(x, \tG_{w,x}^{-1}(y)\bigr) = \frac{\Jac_{\tG_{-w,y}^{-1}}(x)}{\Jac_{\tG_{w,x}^{-1}}(y)}.
\end{equation}
The proof of \Cref{lem:altern_accfun} is concluded by plugging~\eqref{eq:second-identity-altern} and~\eqref{eq:relation-jacobian-altern} into~\eqref{eq:key-identity-altern}.

\subsection{Lifted acceptance probability with deterministic proposals}
\label{subsec:proof_lifted_deterministic}
In this case $\transflifdet(x,p,v)= (\transflif^v(x,p),v)$, $(x,p) \in \rset^{2d}$, $s \in \msv$. Clearly, $\transflifdet^{-1}(x,p,v)= (\transflif^{-v}(x,p),v)$ and it is easily checked that
$\transflifdet^{-1}= \invol \circ \transflifdet \circ \invol$. Denote
\begin{equation}
\label{eq:definition-pi}
  \pi\bigl(\rmd(x,p,v)\bigr) = \pi_0(x) \varphi(p) \rmd x \rmd p \rho(\rmd v) =  \mu(x,p) \rmd x \rmd p \rho(\rmd v) \eqsp.
\end{equation}
To compute the acceptance probability~\eqref{eq:generalized-MH-deterministic}, we need to evaluate the density $k(z)= \rmd \mu/ \rmd \lambda(z)$, where
\begin{equation}
\label{eq:definition-lambda}
\lambda= \pi + \pushf{\transflifdet^{-1}}{\pi}.
\end{equation}
Let $f\colon \rset^d \times \rset^d \times \msv \to \rset_+$ be a measurable function.
We get
\begin{align}
\lambda(f)
&= \int f(x,p,v) \mu(x,p) \rmd x \rmd p \rho(\rmd v) + \int f(\transflif^{-v}(x,p),v) \mu(x,p) \rmd x \rmd p \rho(\rmd v) \\
&= \int f(x,p,v) \mu(x,p) \rmd x \rmd p \rho(\rmd v) + \int f(x',p',v) \mu\bigl(\transflif^{v}(x',p')\bigr) \Jac_{\transflif^v}(x',p')  \rmd x' \rmd p' \rho(\rmd v) \eqsp.
\end{align}
Therefore, we get
\begin{equation}
\frac{\rmd \lambda}{\rmd \Leb_{2d} \otimes \rho}(x,p,v)= \mu(x,p) +
\mu\bigl(\transflif^{v}(x,p)\bigr) \Jac_{\transflif^v}(x,p) \eqsp.
\end{equation}
This implies that, for all $(x,v) \in \rset^d \times \msv$,
\begin{equation}
\label{eq:expression-k-deterministic}
k(x,p,v)= \frac{\mu(x,p)}{\mu(x,p)+ \mu\bigl(\transflif^v(x,p)\bigr)\Jac_{\transflif^v}(x,p)} \eqsp.
\end{equation}
Since  $\invol \circ \transflifdet (x,p,v)= (\transflif^v(x,p),-v)$, we obtain
\begin{align}
k\bigl(\invol \circ \transflifdet(x,p,v)\bigr)&= \frac{\mu\bigl(\transflif^v(x,p)\bigr)}{\mu\bigl(\transflif^v(x,p)\bigr) + \mu(x,p) \Jac_{\transflif^{-v}}\bigl(\transflif^v(x,p)\bigr)} \\
&= \frac{ \mu\bigl(\transflif^v(x,p)\bigr) \Jac_{\transflif^v}(x,p)}{\mu(x,p) + \mu\bigl(\transflif^v(x,p)\bigr) \Jac_{\transflif^v}(x,p)} \eqsp,
\end{align}
where we have used $\Jac_{\transflif^v}(x,p) = 1/ \Jac_{\transflif^{-v}}\bigl(\transflif^v(x,p)\bigr)$. Therefore, the acceptance probability is given by
\begin{equation}
\label{eq:acceptance-probability-determinist-lif-sup}
\alpha\bigl((x,v),(y,w)\bigr)=
\begin{cases}
\accfun\left( \frac{\mu\bigl(\transflif^v(x,p)\bigr) \Jac_{\transflif^v}(x,p)}{\mu(x,p)} \right) \, , \text{ if } \mu(x,p) > 0, (y,w)= (\transflif^v(x,p),v), \\
1,  \, \text{ if } \mu(x,p)=0 \, \text{ or } \, (y,w) \ne (\transflif^v(x,p),v)  \eqsp.
\end{cases}
\end{equation}

\subsection{L2HMC Algorithms}
\label{subsec:l2hmc_implem}
In this section, we discuss the sampling algorithms associated to the L2HMC kernel, \Cref{ex:l2hmc}.
Again, we first describe a version of this algorithm in which the momentum is fully refreshed at each iteration. 
This is a lifted version of the original L2HMC algorithm \cite{levy:hoffman:sohl-dickstein:2017}, because we keep the direction variable at each iteration instead of refreshing the direction at each iteration. 
Consider the following assumption. 
\begin{assumptionL2}
\label{assumption:G_x_diffeo}
For all $v,x \in \{-1,+1\} \times \rset^d$, 
\begin{equation}
\label{eqS:definition-mapping-G}
G_{v,x} \colon p \mapsto  \projq \circ \transflif^v(x,p)
\end{equation}
is a $C^1$-diffeomorphism.
\end{assumptionL2} 
As in the NICE case, establishing \Cref{assumption:G_x_diffeo} requires conditions on the mapping $\Psi$ defining the L2MHC transitions, which is subject to a future work.
Under \Cref{assumption:G_x_diffeo}, \Cref{lem:altern_accfun} shows that the lifted L2HMC with full momentum refresh satisfies the assumption of \Cref{propo:lifted_density}. 
We may therefore apply \Cref{theo:irred_and_co} to show convergence of the algorithm in the sense of \Cref{theo:basic-cv-textbook}.

\begin{algorithm}[!ht]
    \caption{Original L2HMC}
      \label{alg:L2HMC_original}
    \begin{algorithmic}
      \STATE {\bfseries Input:} Transformation $\transflif$, acceptance function $\accfun$, unnormalized target $\pi$, density $\varphi$ of momentum $p$, initial point $x_0$, number of steps $N$
      \FOR{$i=0$ {\bfseries to} $N-1$}
        \STATE Refresh momentum $q_i \sim \varphi$ and direction $v_i \sim \Uniform\{-1,1\}$;
        \STATE Compute proposal $(y_{i+1}, q_{i+1}) =\transflif^{v_i}(x_i, q_i)$;
        \STATE Draw $B_i \sim\Ber(a_i)$ where 
        \[a_i = \accfun\left( \frac{\pi(y_{i+1})\varphi(q_{i+1})}{\pi(x_i)\varphi(q_i)}\Jac_{\transflif^{v_i}}(x_i, q_i)\right)\eqsp;
        \]
        \IF{$B_i\equiv 1$}
          \STATE Set $x_{i+1}=y_{i+1}$;
        \ELSE
          \STATE Set $x_{i+1}=x_{i}$;
        \ENDIF 
      \ENDFOR
      \STATE Return $(x_{0:N})$
    \end{algorithmic}
  \end{algorithm}

\begin{algorithm}[!ht]
    \caption{Lifted L2HMC with full momentum refreshment}
      \label{alg:L2HMC_full}
    \begin{algorithmic}
      \STATE {\bfseries Input:} Transformation $\transflif$, acceptance function $\accfun$, unnormalized target $\pi$, density $\varphi$ of momentum $p$, probability of direction refreshment $\partialrefresh$, initial point $x_0$ and initial direction $v_0$, number of steps $N$
      \FOR{$i=0$ {\bfseries to} $N-1$}
         \STATE Draw $R_i\sim\Ber(\partialrefresh)$;
        \IF{$R_i\equiv 1$}
        \STATE Refresh direction $w_i\sim\Uniform\{-1,1\}$;
        \ELSE
        \STATE Keep direction $w_i=v_i$;
        \ENDIF
      	\STATE Refresh momentum $q_i\sim \varphi$;
        \STATE Compute proposal $(y_{i+1}, q_{i+1}) =\transflif^{w_i}(x_i, q_i)$;
        \STATE Draw $B_i \sim\Ber(a_i)$ where
        \[a_i = \accfun\left( \frac{\pi(y_{i+1})\varphi(q_{i+1})}{\pi(x_i)\varphi(q_i)}\Jac_{\transflif^{w_i}}(x_i, q_i)\right)\eqsp;
        \]
        \IF{$B_i\equiv 1$}
          \STATE Set $(x_{i+1}, v_{i+1})= (y_{i+1}, w_i)$; \hfill\emph{\#\#\# accept the move and keep the direction}
        \ELSE
          \STATE Set $(x_{i+1}, v_{i+1})= (x_{i},  -w_i)$;\hfill{ \emph{\#\#\#reject the move and flip the direction} }
        \ENDIF
      \ENDFOR
      \STATE Return $(x_{0:N})$
    \end{algorithmic}
  \end{algorithm}

As said above, the original L2HMC algorithm (Algorithm~\ref{alg:L2HMC_original}) refreshes at each iteration both the direction and the momentum, whereas the lifted algorithm keeps the direction and refreshes only the momentum. Similar to the NICE case, 
define the marginal kernel, acting on the position only:
\[
  K(x,\rmd y)= K_\alpha(x,\rmd y) + \bigl(1- \bar{\alpha}(x)\bigr) \updelta_x(\rmd y) \eqsp,
\]
where $\bar{\alpha}(x)= K_\alpha(x,\rset^d)$ and for a measurable nonnegative function $f$,
\[
  K_\alpha f(x)= \iint f\bigl(\projq \circ \transflif^v(x,p)\bigr) \bar{\alpha}(x,p,v)  Q\bigl((x,p,v),\rmd(y,q,w)\bigr) \varphi(p) \rmd p \rho (\rmd v) \eqsp,
\]
where 
\begin{equation}
\label{eqS:transition-kernel-deterministic}
  Q\bigl((x,p,v),\rmd(y,q,w)\bigr) = \updelta_{\transflif^v(x,p)}\bigl(\rmd (y,q)\bigr) \updelta_{v}(\rmd w) \eqsp.
\end{equation}
and
\begin{equation}
\label{eqS:acceptance-probability-determinist-lif}
\bar{\alpha}(x,p,v) = \accfun\left( \mu\bigl(\transflif^v(x,p)\bigr) /{\mu(x,p) \;  \Jac_{\transflif^v}(x,p)} \right) \eqsp,
\end{equation}
Recall that the Markov kernel $Q_\alpha\bigl(x,p,v; \rmd( y,q,w)\bigr)= \bar{\alpha}(x,p,v)   Q\bigl((x,p,v),\rmd(y,q,w)\bigr)$ is $(\pi,\involk)$-reversible. Let $f$ and $g$ be two positive measurable functions on $\rset^d$. Since  $Q_\alpha$ is $(\pi,\involk)$-reversible and $\rho$ is symmetric, we get
\begin{align}
\label{eq:reversibility-K-alpha-1}
\int \target(\rmd x) K_\alpha f(x) g(x)
&= \int \pi\bigl(\rmd (x,p,v)\bigr) Q_\alpha\bigl((x,p,v);\rmd(y,q,w)\bigr) f(y) g(x)  \\
&= \int \pi\bigl(\rmd (y,q,w)\bigr) \involk Q_\alpha \involk g(y) f(y) \\
&=^{(3)} \int \pi\bigl(\rmd (y,q,w)\bigr) \bar{\alpha}(y,q,-w) Q\bigl((y,q,-w); \rmd(x,p,v)\bigr) g(x) f(y) \\
&=^{(4)} \int \pi\bigl(\rmd (y,q,w)\bigr) \bar{\alpha}(y,q,w) Q\bigl((y,q,w); \rmd(x,p,v)\bigr) g(x) f(y) \\
&=^{(5)} \int \target(\rmd x) K_\alpha g(y) f(y)
\end{align} 
where we have used in $(3)$ that $\involk Q_\alpha \involk g(y,q,w)= \bar{\alpha}(y,q,-w) \int Q\bigl((y,q,-w); \rmd(x,p,v)\bigr) g(x)$, in $(4)$ the symmetry of $\rho$ and finally in $(5)$ the definition of $K_\alpha$. Hence the 
L2HMC kernel is $\target$-reversible. 
Moreover, note that under \Cref{assumption:G_x_diffeo}, we obtain
\begin{equation}
\label{eqS:MH-expression-K-alpha}
K_\alpha f(x)= \int \left\{ \int \bar{\alpha}(x,G_{v,x}^{-1}(y),v) \varphi\bigl(G_{v,x}^{-1}(y)\bigr) \Jac_{G_{v,x}^{-1}}(y) \rho(\rmd v) \right\} f(y) \rmd y \eqsp.
\end{equation}
Denoting by $q_v(x,y)$ the transition density $q_v(x,y) = \varphi\bigl(G_{v,x}^{-1}(y)\bigr) \Jac_{G_{v,x}^{-1}}(y) $ and then setting  $q(x,y)= \int q_v(x,y) \rho(\rmd v)$, we finally get 
\begin{equation}
\label{eqS:MH-expression-K-alpha}
K_\alpha f(x) = \int \alpha(x,y) q(x,y) f(y) \rmd y 
\end{equation}
with
\begin{equation}
\label{eqS:acceptance-ratio-L2HMC}
\alpha(x,y)= \frac{\int \bar{\alpha}(x,G_{v,x}^{-1}(y),v) \varphi\bigl(G_{v,x}^{-1}(y)\bigr) \Jac_{G_{v,x}^{-1}}(y) \rho(\rmd v)}{\int \varphi\bigl(G_{v,x}^{-1}(y)\bigr) \Jac_{G_{v,x}^{-1}}(y) \rho(\rmd v)}
\end{equation}
Write now, following~\citep{tierney:1994}, 
\begin{equation}
r(x,y) = \frac{\target(x) q(x,y)}{\target(y)q(y,x)}\eqsp.
\end{equation}
Moreover, by \Cref{lem:altern_accfun}, we can write 
\begin{equation}
\bar{\alpha}(x,G_{v,x}^{-1}(y),v) = \accfun\left(\frac{\target(y)q_{-v}(y,x)}{\target(x)q_v(x,y)} \right)\eqsp.
\end{equation}
In that case, we have
\begin{align}
\alpha(x,y) \target(x) q(x,y) &= \int \target(x) q_v(x,y) \accfun\left(\frac{\target(y)q_{-v}(y,x)}{\target(x)q_v(x,y)} \right) \rho(\rmd v)\\
&= \int \target(y) q_{-v}(y,x) \accfun\left(\frac{\target(x)q_v(x,y)}{\target(y)q_{-v}(y,x)} \right) \rho(\rmd v)\\
&= \int \target(y) q_{w}(y,x) \accfun\left(\frac{\target(x)q_{-w}(x,y)}{\target(y)q_{w}(y,x)} \right) \rho(\rmd w)\eqsp,
\end{align}
where we have use the fact that $t\accfun(1/t) =  \accfun(t)$ and the change of variable $w=-v$.
Then, 
\begin{equation}
\alpha(x,y) \target(x) q(x,y) = \alpha(y,x) \target(y) q(y,x)\eqsp.
\end{equation}
We thus have $\alpha(x,y) r(x,y) = \alpha(y,x)$, which proves by \citep[Theorem 2]{tierney:1994} that the ratio $\alpha$ is exactly the classical MH ratio satisfying the detailed balance condition.

To retrieve persistency, we can use as for  the NICE algorithm a mixture of a deterministic L2HMC move and a full independent refreshment of the position, momentum and the direction; see Algorithm~\ref{alg:L2HMC_lifted}.
 \newpage
% Moreover, again, we display an alternative version of the persistency similar to \citep{horowitz:1991}, in Algorithm~\ref{alg:L2HMC_persistent}.

\begin{algorithm}[!ht]
    \caption{Lifted L2HMC with randomized full refreshment}
     \label{alg:L2HMC_lifted}
 \begin{algorithmic}
      \STATE {\bfseries Input:} Transformation $\transflif$, acceptance function $\accfun$, unnormalized target $\pi$, density $\varphi$ of momentum $p$, probability of refreshment $\partialrefresh$, initial point $x_0$, initial momentum $p_0$ and initial direction $v_0$, number of steps $N$
      \FOR{$i=0$ {\bfseries to} $N-1$}
      \STATE Draw $R_i\sim \Ber(\partialrefresh)$;
      \IF{$R_i\equiv 0$}
        \STATE Compute proposal $(y_{i+1}, q_{i+1}) =\transflif^{v_i}(x_i, p_i)$; \emph{\#\#\#No refreshment, deterministic dynamics}
        \STATE Draw $B_i\sim\Ber(a_i)$ where
        \[a_i = \accfun\left( \frac{\pi(y_{i+1})\varphi(q_{i+1})}{\pi(x_i)\varphi(p_i)}\Jac_{\transflif^{v_i}}(x_i, p_i)\right)\eqsp;
        \]
        \IF{$B_i\equiv 1$}
          \STATE Set $(x_{i+1},p_{i+1}, v_{i+1})= (y_{i+1}, q_{i+1}, v_i)$; \hfill{ \emph{\#\#\# accept the move and keep the direction}}
        \ELSE
          \STATE Set $(x_{i+1}, p_{i+1}, v_{i+1})= (x_{i}, p_i, -v_i)$; \hfill{ \emph{\#\#\#reject the move and flip the direction} }
         \ENDIF
      \ELSE
%       \STATE\hfill\emph{\#\#\# Full refresh }
        \STATE Draw $q_i \sim \varphi$, $w_i \sim \Uniform\{-1,1\}$; \hfill \emph{\#\#\# refresh independently the momentum and the direction}% to update the position}
        \STATE Compute proposal $(y_{i+1}, q_{i+1}) =\transflif^{w_i}(x_i, q_i)$;
        \STATE Draw $B_i\sim\Ber(a_i)$ where
        \[a_i = \accfun\left( \frac{\pi(y_{i+1})\varphi(q_{i+1})}{\pi(x_i)\varphi(q_i)}\Jac_{\transflif^{w_i}}(x_i, q_i)\right)\eqsp;
        \]
        \STATE Draw $p_{i+1} \sim \varphi$, $v_{i+1} \sim \Uniform\{-1,1\}$; \hfill \emph{\#\#\# refresh the momentum and the direction}
        \IF{$B_i\equiv 1$}
          \STATE Set $x_{i+1} = y_{i+1}$;
        \ELSE
          \STATE Set $x_{i+1} = x_{i}$;
        \ENDIF
      \ENDIF
      \ENDFOR
      \STATE Return $(x_{0:N})$
    \end{algorithmic}
  \end{algorithm}
Much like the persistent HMC algorithm, we may also design a lifted persistent HMC algorithm in which, at each iteration, we keep the direction and partially refresh the momentum using an autoregressive scheme; see Algorithm~\ref{alg:L2HMC_persistent}. 
\newpage

 \begin{algorithm}[!ht]
    \caption{Lifted L2HMC with persistence}
    \label{alg:L2HMC_persistent}
    \begin{algorithmic}
      \STATE {\bfseries Input:} Transformation $\transflif$, acceptance function $\accfun$, unnormalized target $\pi$, density $\varphi$ of momentum $p$, hyperparameter $\beta$, initial point $x_0$, initial momentum $p_0$ and initial direction $v_0$, number of steps $N$
      \FOR{$i=0$ {\bfseries to} $N-1$}
      	\STATE Sample $u_i\sim \varphi$ and refresh momentum $q_i = \beta p_i + \sqrt{1- \beta^2} u_i$; \hfill \emph{\#\#\# partially update the momentum}
        \STATE Compute proposal $(y_{i+1}, q_{i+1}) =\transflif^{v_i}(x_i, q_i)$; 
        \STATE Draw $B_i\sim\Ber(a_i)$ where
        \[a_i = \accfun\left( \frac{\pi(y_{i+1})\varphi(q_{i+1})}{\pi(x_i)\varphi(q_i)}\Jac_{\transflif^{v_i}}(x_i, q_i)\right)\eqsp;
        \]
        \IF{$B_i\equiv 1$}
          \STATE Set $(x_{i+1},p_{i+1}, v_{i+1})= (y_{i+1}, q_{i+1}, v_i)$; \hfill \emph{\#\#\# accept the move and keep the direction}
        \ELSE
          \STATE Set $(x_{i+1}, p_{i+1}, v_{i+1})= (x_{i}, p_i, -v_i)$; \hfill \emph{\#\#\# reject the move and flip the direction}
        \ENDIF
      \ENDFOR
      \STATE Return $(x_{0:N})$
    \end{algorithmic}
  \end{algorithm}

\end{document}